\documentclass[a4paper,11pt]{article}
\pdfoutput=1

\usepackage[hmargin=1.5in,vmargin=1in]{geometry}
\usepackage{hyperref,microtype}
\usepackage{amsbsy,amsfonts,amsmath,amssymb,amsthm,booktabs,enumitem,graphicx,hyphenat,titlesec}
\usepackage[all]{xy}
\usepackage[square,numbers,sort&compress]{natbib}
\usepackage[font=small,margin=2em]{caption}

\theoremstyle{plain}
\newtheorem{thm}{Theorem}
\newtheorem{prop}[thm]{Proposition}

\newtheorem{lemma}[thm]{Lemma}

\theoremstyle{definition}
\newtheorem{definition}{Definition}

\theoremstyle{remark}
\newtheorem{remark}{Remark}[section]

\newcommand{\Apx}[1]{$#1$-approx\-i\-mation}
\newcommand{\floor}[1]{\lfloor #1 \rfloor}
\newcommand{\<}{\mathord{<}}
\newcommand{\even}[1]{\ddot{#1}}
\newcommand{\odd}[1]{\dot{#1}}

\newcommand{\N}{\ensuremath{\mathbb{N}}}
\newcommand{\Z}{\ensuremath{\mathbb{Z}}}
\newcommand{\R}{\ensuremath{\mathbb{R}}}

\newcommand{\Alg}{\ensuremath{\mathsf{A}}}
\newcommand{\Blg}{\ensuremath{\mathsf{B}}}
\newcommand{\Vlg}{\ensuremath{\mathsf{V}}}
\newcommand{\PI}{\ensuremath{\mathsf{\Pi}}}

\newcommand{\Gr}{\ensuremath{\mathcal{G}}}
\newcommand{\Hr}{\ensuremath{\mathcal{H}}}
\newcommand{\Cr}{\ensuremath{\mathcal{C}}}
\newcommand{\Tr}{\ensuremath{\mathcal{T}}}
\newcommand{\Wr}{\ensuremath{\mathcal{W}}}
\newcommand{\Ur}{\ensuremath{\mathcal{U}}}
\newcommand{\Cay}{\ensuremath{\mathcal{C}}}

\newcommand{\Ff}{\ensuremath{\mathcal{F}}}
\newcommand{\Wf}{\ensuremath{\mathfrak{W}}}

\newcommand{\PO}{\mathsf{PO}}
\newcommand{\OI}{\mathsf{OI}}
\newcommand{\ID}{\mathsf{ID}}
\newcommand{\LCL}{\mathsf{LCL}}

\newcounter{myexternalpagenum}
\newcommand{\definepage}[1]{\stepcounter{myexternalpagenum}\edef#1{\arabic{myexternalpagenum}}}
\definepage{\PModels}
\definepage{\PCycles}
\definepage{\PLifts}
\definepage{\PLDigraph}
\definepage{\PComplete}
\definepage{\PHomog}
\definepage{\PHomogTree}
\definepage{\PProduct}
\definepage{\PHighGirth}

\newcommand{\mydoi}[2]{\href{http://dx.doi.org/#1}{#2}}
\newcommand{\myeprint}[2]{\href{http://arxiv.org/abs/#1}{#2}}
\newcommand{\myurl}[2]{\href{#1}{#2}}

\titleformat{\section}          {\raggedright\normalfont\Large\bfseries}{\thesection}{1em}{}
\titleformat{\subsection}       {\raggedright\normalfont\large\bfseries}{\thesubsection}{1em}{}

\setitemize{label=--}
\setenumerate{label=(\alph*)}

\hypersetup{
    colorlinks=true,
    linkcolor=black,
    citecolor=black,
    filecolor=black,
    urlcolor=black,
    pdftitle={Lower bounds for local approximation},
    pdfauthor={Mika G{\"o}{\"o}s, Juho Hirvonen, Jukka Suomela},
}

\begin{document}

\vspace*{5ex}
\begin{center}
    {\Large \textbf{Lower Bounds for Local Approximation}}

    \bigskip
    \bigskip

    \href{http://www.cs.helsinki.fi/people/mgoos}{Mika G{\"o}{\"o}s},
    \href{http://www.cs.helsinki.fi/people/jghirvon}{Juho Hirvonen}, and
    \href{http://www.cs.helsinki.fi/u/josuomel/}{Jukka Suomela}

    \bigskip
    \href{http://www.hiit.fi/}{Helsinki Institute for Information Technology HIIT} \\
    \href{http://www.cs.helsinki.fi/en}{Department of Computer Science} \\
    \href{http://www.helsinki.fi/university/}{University of Helsinki}, Finland
    
    \bigskip
    \{mika.goos, juho.hirvonen, jukka.suomela\}@cs.helsinki.fi
\end{center}

\bigskip
\bigskip
\noindent\textbf{Abstract.}
In the study of deterministic distributed algorithms it is commonly assumed that each node has a unique $O(\log n)$-bit identifier. We prove that for a general class of graph problems, local algorithms (constant-time distributed algorithms) do not need such identifiers: a port numbering and orientation is sufficient.

Our result holds for so-called \emph{simple $\PO$-checkable graph optimisation problems}; this includes many classical packing and covering problems such as vertex covers, edge covers, matchings, independent sets, dominating sets, and edge dominating sets. We focus on the case of bounded-degree graphs and show that if a local algorithm finds a constant-factor approximation of a simple $\PO$-checkable graph problem with the help of unique identifiers, then the same approximation ratio can be achieved on anonymous networks.

As a corollary of our result and by prior work, we derive a tight lower bound on the local approximability of the \emph{minimum edge dominating set problem}.

Our main technical tool is an algebraic construction of \emph{homogeneously ordered graphs}: We say that a graph is $(\alpha,r)$-homogeneous if its nodes are linearly ordered so that an $\alpha$ fraction of nodes have pairwise isomorphic radius-$r$ neighbourhoods. We show that there exists a finite $(\alpha,r)$-homogeneous $2k$-regular graph of girth at least $g$ for any $\alpha < 1$ and any $r$, $k$, and $g$.

\thispagestyle{empty}
\setcounter{page}{0}
\newpage

\section{Introduction}

In this work, we study deterministic distributed algorithms under three different assumptions; see Figure~\ref{fig:models}.
\begin{itemize}[leftmargin=4em]
    \item[($\ID$)] \emph{Networks with unique identifiers}. Each node is given a unique $O(\log n)$-bit label.
    \item[($\OI$)] \emph{Order-invariant algorithms}. There is a linear order on nodes. Equivalently, the nodes have unique labels, but the output of an algorithm is not allowed to change if we relabel the nodes while preserving the relative order of the labels.
    \item[($\PO$)] \emph{Anonymous networks with a port numbering and orientation}. For each node, there is a linear order on the incident edges, and for each edge, there is a linear order on the incident nodes. Equivalently, a node of degree $d$ can refer to its neighbours by integers $1, 2, \dotsc, d$, and each edge is oriented so that the endpoints know which of them is the head and which is the tail.
\end{itemize}

\begin{figure}
    \centering
    \includegraphics[page=\PModels]{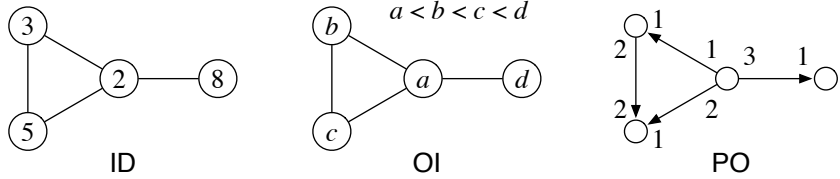}
    \caption{Three models of distributed computing.}\label{fig:models}
\end{figure}

While unique identifiers are often useful, we will show that they are seldom needed in local algorithms (constant-time distributed algorithms): there is a general class of graph problems such that local algorithms in $\PO$ are able to produce as good approximations as local algorithms in $\OI$ or~$\ID$.

\subsection{Graph Problems}

We study graph problems that are related to the structure of an unknown communication network. Each node in the network is a computer; each computer receives a \emph{local input}, it can exchange messages with adjacent nodes, and eventually it has to produce a \emph{local output}. The local outputs constitute a solution of a graph problem---for example, if we study the dominating set problem, each node produces one bit of local output, indicating whether it is part of the dominating set. The \emph{running time} of an algorithm is the number of synchronous communication rounds.

From this perspective, the models $\ID$, $\OI$, and $\PO$ are easy to separate. Consider, for example, the problem of finding a maximal independent set in an $n$-cycle. In $\ID$ model the problem can be solved in $\Theta(\log^* n)$ rounds \cite{cole86deterministic, linial92locality}, while in $\OI$ model we need $\Theta(n)$ rounds, and the problem is not soluble at all in $\PO$, as we cannot break symmetry---see Figure~\ref{fig:cycles}. Hence $\ID$ is strictly stronger than $\OI$, which is strictly stronger than $\PO$.

\begin{figure}
    \centering
    \includegraphics[page=\PCycles]{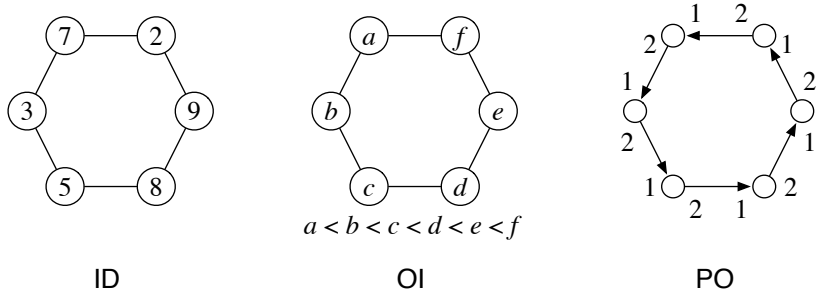}
    \caption{In $\ID$, the numerical identifiers break symmetry everywhere---for example, a maximal independent set can be found in $O(\log^* n)$ rounds. In $\OI$, we can have a cycle with only one ``seam'', and in $\PO$ we can have a completely symmetric cycle.}\label{fig:cycles}
\end{figure}

\subsection{Local Algorithms}\label{ssec:local}

In this work we focus on \emph{local algorithms}, i.e., distributed algorithms that run in a constant number of synchronous communication rounds, independently of the number of nodes in the network~\cite{naor95what, suomela09survey}. The above example separating $\ID$, $\OI$, and $\PO$ no longer applies, and there has been a conspicuous lack of \emph{natural} graph problems that would separate $\ID$, $\OI$, and $\PO$ from the perspective of local algorithms.

Indeed, there are results that show that many problems that can be solved with a local algorithm in $\ID$ also admit a local algorithm in $\OI$ or $\PO$. For example, the seminal paper by Naor and Stockmeyer~\cite{naor95what} studies so-called $\LCL$ problems---these include problems such as graph colouring and maximal matchings on bounded-degree graphs. The authors show that $\ID$ and $\OI$ are indeed equally expressive among $\LCL$ problems. The followup work by Mayer, Naor, and Stockmeyer~\cite{mayer95local} hints of a stronger property:
\begin{enumerate}[label=(\roman*)]
    \item \emph{Weak $2$-colouring} is an $\LCL$ problem that can be solved with a local algorithm in $\ID$ model~\cite{naor95what}. It turns out that the same problem can be solved in $\PO$ model as well~\cite{mayer95local}.
\end{enumerate}
Granted, contrived counterexamples do exist: there are $\LCL$ problems that are soluble in $\OI$ but not in $\PO$. However, most of the classical graph problems that are studied in the field of distributed computing are \emph{optimisation problems}, not $\LCL$ problems.

\subsection{Local Approximation}\label{ssec:intro-local-apx}

In what follows, we will focus on graph problems in the case of \emph{bounded-degree graphs}; that is, there is a known constant $\Delta$ such that the degree of any node in any graph that we may encounter is at most~$\Delta$. Parity often matters; hence we also define $\Delta' = 2 \floor{\Delta/2}$.

In this setting, the best possible approximation ratios are surprisingly similar in $\ID$, $\OI$, and $\PO$. The following hold for any given $\Delta \ge 2$ and $\epsilon > 0$:
\begin{enumerate}[resume*]
    \item \emph{Minimum vertex cover} can be approximated to within factor $2$ in each of these models \cite{astrand09vc2apx, astrand10vc-sc}. This is tight: \Apx{(2-\epsilon)} is not possible in any of these models \cite{czygrinow08fast, lenzen08leveraging, suomela09survey}.
    \item \emph{Minimum edge cover} can be approximated to within factor $2$ in each of these models \cite{suomela09survey}. This is tight: \Apx{(2-\epsilon)} is not possible in any of these models \cite{czygrinow08fast, lenzen08leveraging, suomela09survey}.
    \item \emph{Minimum dominating set} can be approximated to within factor $\Delta'+1$ in each of these models \cite{astrand10weakly-coloured}. This is tight: \Apx{(\Delta'+1-\epsilon)} is not possible in any of these models \cite{czygrinow08fast, lenzen08leveraging, suomela09survey}.
    \item \emph{Maximum independent set} cannot be approximated to within any constant factor in any of these models \cite{czygrinow08fast, lenzen08leveraging}.
    \item \emph{Maximum matching} cannot be approximated to within any constant factor in any of these models \cite{czygrinow08fast, lenzen08leveraging}.
\end{enumerate}
This phenomenon has not been fully understood: while there are many problems with identical approximability results for $\ID$, $\OI$, and $\PO$, it has not been known whether these are examples of a more general principle or merely isolated coincidences. In fact, for some problems, tight approximability results have been lacking for $\ID$ and $\OI$, even though tight results are known for $\PO$:
\begin{enumerate}[resume*]
    \item \emph{Minimum edge dominating set} can be approximated to within factor $4-2/\Delta'$ in each of these models \cite{suomela10eds}. This is tight for $\PO$ but only near-tight for $\ID$ and $\OI$: \Apx{(4-2/\Delta'-\epsilon)} is not possible in $\PO$ \cite{suomela10eds}, and \Apx{(3-\epsilon)} is not possible in $\ID$ and $\OI$ \cite{czygrinow08fast, lenzen08leveraging, suomela09survey}.
\end{enumerate}

In this work we prove a theorem unifying all of the above observations---they are indeed examples of a general principle. As a simple application of our result, we settle the local approximability of the minimum edge dominating set problem by proving a tight lower bound in $\ID$ and $\OI$.

\subsection{Main Result}

A \emph{simple graph problem} $\PI$ is an optimisation problem in which a feasible solution is a subset of nodes or a subset of edges, and the goal is to either minimise or maximise the size of a feasible solution. We say that $\PI$ is a \emph{$\PO$-checkable graph problem} if there is a local $\PO$-algorithm $\Alg$ that recognises a feasible solution. That is, $\Alg(\Gr,X,v) = 1$ for all nodes $v \in V(\Gr)$ if $X$ is a feasible solution of problem $\PI$ in graph $\Gr$, and $\Alg(\Gr,X,v) = 0$ for some node $v \in V(\Gr)$ otherwise---here $\Alg(\Gr,X,v)$ is the output of a node $v$ if we run algorithm $\Alg$ on graph $\Gr$ and the local inputs form an encoding of $X$.

Let $\varphi\colon V(\Hr) \to V(\Gr)$ be a surjective graph homomorphism from graph $\Hr$ to graph~$\Gr$. If $\varphi$ preserves vertex degrees, i.e., $\deg_{\Hr}(u)=\deg_{\Gr}(\varphi(u))$, then $\varphi$ is called a \emph{covering map}, and $\Hr$ is said to be a \emph{lift} of $\Gr$. The \emph{fibre} of $u\in V(\Gr)$ is the set $\varphi^{-1}(u)$ of pre-images of $u$. We usually consider $n$-lifts that have fibres of the same cardinality $n$. It is a basic fact that a connected lift $\Hr$ of $\Gr$ is an $n$-lift for some $n$. See Figure~\ref{fig:lifts} for an illustration.

\begin{figure}
    \centering
    \includegraphics[page=\PLifts]{figs.pdf}
    \caption{Graph $\Hr$ is a lift of $\Gr$. The covering map $\varphi\colon V(\Hr) \to V(\Gr)$ maps $a_i \mapsto a$, $b_i \mapsto b$, $c_i \mapsto c$, and $d_i \mapsto d$ for each $i = 1, 2$. The fibre of $a \in V(\Gr)$ is $\{a_1, a_2\} \subseteq V(\Hr)$; all fibres have the same size.}\label{fig:lifts}
\end{figure}

Let $\Ff$ be a family of graphs. We say that $\Ff$ is \emph{closed under lifts} if $\Gr \in \Ff$ implies $\Hr \in \Ff$ for all lifts $\Hr$ of $\Gr$. A family is \emph{closed under connected lifts} if $\Gr \in \Ff$ implies $\Hr \in \Ff$ whenever $\Hr$ and $\Gr$ are connected graphs and $\Hr$ is a lift of~$\Gr$.

\pagebreak 

Now we are ready to state our main theorem.
\begin{thm}[Main Theorem]\label{thm:main}
Let $\PI$ be a simple $\PO$-checkable graph problem. Assume one of the following:
\begin{itemize}
    \item General version: $\Ff$ is a family of bounded degree graphs, and it is closed under lifts.
    \item Connected version: $\Ff$ is a family of connected bounded degree graphs, it does not contain any trees, and it is closed under connected lifts.
\end{itemize}
If there is a local $\ID$-algorithm $\Alg$ that finds an \Apx{\alpha} of $\PI$ in $\Ff$, then there is a local $\PO$-algorithm $\Blg$ that finds an \Apx{\alpha} of $\PI$ in $\Ff$.
\end{thm}

While the definitions are somewhat technical, it is easy to verify that the result is widely applicable:
\begin{enumerate}
    \item Vertex covers, edge covers, matchings, independent sets, dominating sets, and edge dominating sets are simple $\PO$-checkable graph problems.
    \item Bounded-degree graphs, regular graphs, and cyclic graphs are closed under lifts.
    \item Connected bounded-degree graphs, connected regular graphs, and connected cyclic graphs are closed under connected lifts. 
\end{enumerate}

\subsection{An Application}

The above result provides us with a powerful tool for proving lower-bound results: we can easily transfer negative results from $\PO$ to $\OI$ and $\ID$. We demonstrate this strength by deriving a new lower bound result for the minimum edge dominating set problem.
\begin{thm}\label{thm:eds}
    Let $\Delta \ge 2$, and let $\Alg$ be a local $\ID$-algorithm that finds an \Apx{\alpha} of a minimum edge dominating set on connected graphs of maximum degree $\Delta$. Then $\alpha \ge \alpha_0$, where
    \[
        \alpha_0 = 4-2/\Delta' \ \text{ and }\ 
        \Delta' = 2 \floor{\Delta/2}.
    \]
    This is tight: there is a local $\ID$-algorithm that finds an \Apx{\alpha_0}.
\end{thm}
\begin{proof}
    By prior work~\cite{suomela10eds}, it is known that there is a connected $\Delta'$-regular graph $\Gr_0$ such that the approximation factor of any local $\PO$-algorithm on $\Gr_0$ is at least $\alpha_0$. Let $\Ff_0$ consist of all connected lifts of $\Gr_0$, and let $\Ff$ consist of all connected graphs of degree at most $\Delta$. We make the following observations.
    \begin{enumerate}
        \item We have $\Ff_0 \subseteq \Ff$; by assumption, $\Alg$ finds an \Apx{\alpha} in $\Ff_0$.
        \item Family $\Ff_0$ consists of connected graphs of degree at most $\Delta$, it does not contain any trees, and it is closed under connected lifts. We can apply the connected version of the main theorem: there is a local $\PO$-algorithm $\Blg$ that finds an \Apx{\alpha} in $\Ff_0$.
        \item However, $\Gr_0 \in \Ff_0$, and hence $\alpha \ge \alpha_0$.
    \end{enumerate}
    The matching upper bound is presented in prior work~\cite{suomela10eds}.
\end{proof}

\subsection{Overview}

Informally, our proof of the main theorem is structured as follows.
\begin{enumerate}
    \item Fix a graph problem $\PI$, a graph family $\Ff$, and an $\ID$-algorithm $\Alg$ as in the statement of Theorem~\ref{thm:main}. Let $r$ be the running time of $\ID$-algorithm $\Alg$.
    \item Let $\Gr \in \Ff$ be a graph with a port numbering and orientation.
    \item \mbox{Section~\ref{ssec:homog-lift}:} We construct a certain lift $\Gr_\epsilon \in \Ff$ of $\Gr$. Graph $\Gr_\epsilon$ inherits the port numbering and the orientation from $\Gr$.
    \item \mbox{Section~\ref{ssec:mainthm-oi}:} We show that there exists a linear order $<_\epsilon$ on the nodes of $\Gr_\epsilon$ that gives virtually no new information in comparison with the port numbering and orientation. If we have an $\OI$-algorithm $\Alg'$ with running time~$r$, then we can simulate $\Alg'$ with a $\PO$-algorithm $\Blg'$ almost perfectly on $\Gr_\epsilon$: the outputs of $\Alg'$ and $\Blg'$ agree for a $(1-\epsilon)$ fraction of nodes. We deduce that the approximation ratio of $\Alg'$ on $\Ff$ cannot be better than the approximation ratio of $\Blg'$ on~$\Ff$.
    \item \mbox{Section~\ref{ssec:mainthm-id}:} We apply Ramsey's theorem to show that the unique identifiers do not help, either. We can construct a $\PO$-algorithm $\Blg$ that simulates $\Alg$ in the following sense: there exists an assignment of unique identifiers on a lift $\Hr \in \Ff$ of $\Gr_\epsilon$ such that the outputs of $\Alg$ and $\Blg$ agree for a $(1-\epsilon)$ fraction of nodes. We deduce that the approximation ratio of $\Alg$ on $\Ff$ cannot be better than the approximation ratio of $\Blg$ on~$\Ff$.
\end{enumerate}

Now if graph $\Gr$ was a directed cycle, the construction would be standard; see, e.g., Czygrinow et al.~\cite{czygrinow08fast}. In particular, $\Gr_\epsilon$ and $\Hr$ would simply be long cycles, and $<_\epsilon$ would order the nodes along the cycle---there would be only one ``seam'' in $(\Gr_\epsilon,\<_\epsilon)$ that could potentially help $\Alg'$ in comparison with $\Blg'$, and only an $\epsilon$ fraction of nodes are near the seam.

However, the case of a general $\Gr$ is more challenging. Our main technical tool is the construction of so-called homogeneous graphs; see Section~\ref{ssec:homog}. Homogeneous graphs are regular graphs with a linear order that is useless from the perspective of $\OI$-algorithms: for a $(1-\epsilon)$ fraction of nodes, the local neighbourhoods are isomorphic. Homogeneous graphs trivially exist; however, our proof calls for homogeneous graph of an arbitrarily high degree and an arbitrarily large girth (i.e., there are no short cycles---the graph is locally tree-like). In Section~\ref{sec:homog-graphs} we use an algebraic construction to prove that such graphs exist.

\subsection{Discussion}

In the field of distributed algorithms, the running time of an algorithm is typically analysed in terms of two parameters: $n$, the number of nodes in the graph, and $\Delta$, the maximum degree of the graph. In our work, we assumed that $\Delta$ is a constant---put otherwise, our work applies to algorithms that have a running time independent of $n$ but arbitrarily high as a function of $\Delta$. The work by Kuhn et al.~\cite{kuhn04what, kuhn06price, kuhn10local} studies the dependence on $\Delta$ more closely: their lower bounds on approximation ratios apply to algorithms that have, for example, a running time sublogarithmic in $\Delta$.

While our result is very widely applicable, certain extensions have been left for future work. One example is the case of planar graphs \cite{czygrinow08fast},~\cite[\S13]{lenzen11phd}. The family of planar graphs is not closed under lifts, and hence Theorem~\ref{thm:main} does not apply. Another direction that we do not discuss at all is the case of randomised algorithms.

\section{Three Models of Distributed Computing}

In this section we make precise the notion of a \emph{local algorithm} in each of the models $\ID$, $\OI$ and $\PO$. First, we discuss the properties common to all the models.

We start by fixing a graph family $\Ff$ where every $\Gr=(V(\Gr),E(\Gr))\in\Ff$ has maximum degree at most $\Delta\in\N$. We consider algorithms $\Alg$ that operate on graphs in $\Ff$; the properties of $\Alg$ (e.g., its running time) are allowed to depend on the family $\Ff$ (and, hence, on $\Delta$). We denote by $\Alg(\Gr,u)\in\Omega$ the output of $\Alg$ on a node $u\in V(\Gr)$. Here, $\Omega$ is a finite set of possible outputs of $\Alg$ in $\Ff$. If the solutions to $\PI$ are sets of vertices, we shall have $\Omega = \{0,1\}$ so that the solution produced by $\Alg$~on $\Gr$, denoted $\Alg(\Gr)$, is the set of nodes $u$ with $\Alg(\Gr,u)=1$. Similarly, if the solutions to $\PI$ are sets of edges, we shall have $\Omega = \{0,1\}^\Delta$ so that the $i$th component of the vector $\Alg(\Gr,u)$ indicates whether the $i$th edge incident to $u$ is included in the solution $\Alg(\Gr)$---in each of the models a node will have a natural ordering of its incident edges.

Let $r\in \N$ denote the constant running time of $\Alg$ in $\Ff$. This means that a node $u$ can only receive messages from nodes within distance $r$ in $\Gr$, i.e., from nodes in the radius-$r$ neighbourhood
\[
    B_{\Gr}(u,r) = \bigl\{ v\in V(\Gr) : \operatorname{dist}_{\Gr}(u,v) \le r \bigr\}.
\]
Let $\tau(\Gr,u)$ denote the structure $(\Gr,u)$ restricted to the vertices $B_{\Gr}(u,r)$, i.e., in symbols,
\[
    \tau(\Gr,u) = (\Gr,u) \upharpoonright B_{\Gr}(u,r).
\]
Then $\Alg(\Gr,u)$ is a function of the data $\tau(\Gr,u)$ in that $\Alg(\Gr,u) = \Alg(\tau(\Gr,u))$. The models $\ID$, $\OI$ and $\PO$ impose further restrictions on this function.

\subsection{Model \texorpdfstring{$\boldsymbol\ID$}{ID}}\label{sec:id-model}

Local $\ID$-algorithms are not restricted in any additional way. We follow the convention that the vertices have unique $O(\log n)$-bit labels, i.e., an instance $\Gr\in\Ff$ of order $n=|V(\Gr)|$ has $V(\Gr)\subseteq \{1,2,\dotsc,s(n)\}$ where $s(n)$ is some fixed polynomial function of $n$. Our presentation assumes $s(n)=\omega(n)$, even though this assumption can often be relaxed as we discuss in Remark~\ref{rem:identifiers}.

\subsection{Model \texorpdfstring{$\boldsymbol\OI$}{OI}}

A local $\OI$-algorithm $\Alg$ does not directly use unique vertex identifiers but only their relative \emph{order}. To make this notion explicit, let the vertices of $\Gr\in\Ff$ be linearly ordered by $<$, and call $(\Gr,\<)$ an \emph{ordered graph}. Denote by $\tau(\Gr,\<,u)$ the restriction of the structure $(\Gr,\<,u)$ to the $r$-neighbourhood $B_{\Gr}(u,r)$, i.e., in symbols,
\[
    \tau(\Gr,\<,u) = (\Gr,\<,u) \upharpoonright B_{\Gr}(u,r).
\]
Then, the output $\Alg(\Gr,\<,u)$ depends only on the \emph{isomorphism type} of $\tau(\Gr,\<,u)$, so that if $\tau(\Gr,\<,u) \allowbreak\simeq \tau(\Gr',\<',u')$ then $\Alg(\Gr,\<,u) = \Alg(\Gr',\<',u')$.

\subsection{Model \texorpdfstring{$\boldsymbol\PO$}{PO}}\label{sec:po-model}

In the $\PO$ model the nodes are considered anonymous and only the following node specific structure is available: a node can communicate with its neighbours through ports numbered $1,2,\dotsc,\deg(u)$, and each communication link has an orientation.

\paragraph{Edge-Labelled Digraphs.}

To model the above, we consider \emph{$L$-edge-labelled directed graphs} (or \emph{$L$-digraphs}, for short) $\Gr=(V(\Gr),E(\Gr),\ell_{\Gr})$, where the edges $E(\Gr)\subseteq V(\Gr)\times V(\Gr)$ are directed and each edge $e\in E(\Gr)$ carries a label $\ell_{\Gr}(e)\in L$. We restrict our considerations to \emph{proper} labellings $\ell_{\Gr}\colon E(\Gr)\to L${} that for each $u\in V(\Gr)$ assign the incoming edges $(v,u)\in E(\Gr)$ distinct labels and the outgoing edges $(u,w)\in E(\Gr)$ distinct labels; we allow $\ell_{\Gr}(v,u)=\ell_{\Gr}(u,w)$. We refer to the outgoing edges of a node by the labels $L$ and to the incoming edges by the formal letters $L^{-1} = \{\ell^{-1} : \ell\in L\}$. In the context of $L$-digraphs, covering maps $\varphi\colon V(\Hr) \to V(\Gr)$ are required to preserve edge labels so that $\ell_{\Hr}(u,v) = \ell_{\Gr}(\varphi(u), \varphi(v))$ for all $(u,v)\in E(\Hr)$.

A port numbering on $\Gr$ gives rise to a proper labelling $\ell_{\Gr}(v,u) = (i,j)$, where $u$ is the $i$th neighbour of $v$, and $v$ is the $j$th neighbour of $u$; see Figure~\ref{fig:ldigraph}. We now fix $L$ to contain every possible edge label that appears when a graph $\Gr\in\Ff$ is assigned a port numbering and an orientation. Note that $|L| \le \Delta^2$.

\begin{figure}
    \centering
    \includegraphics[page=\PLDigraph]{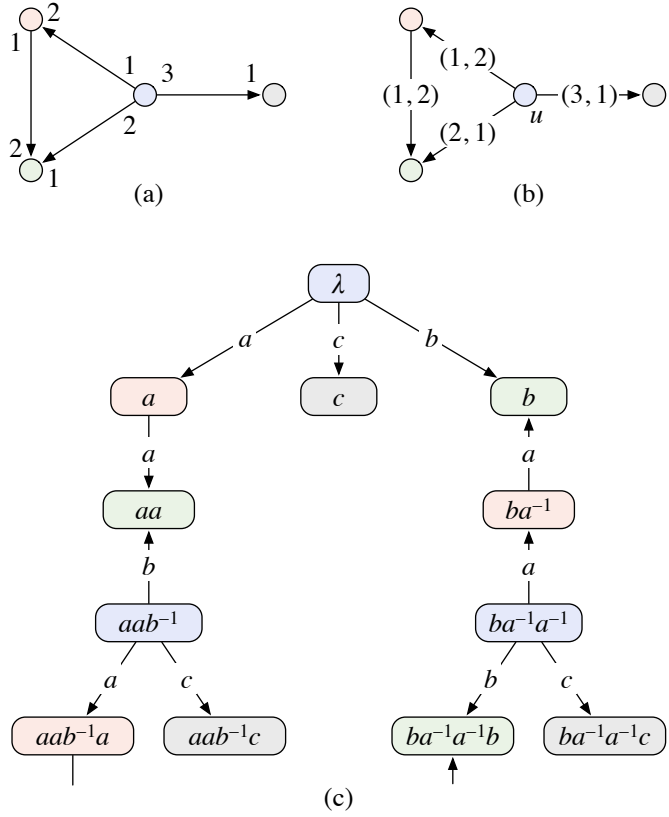}
    \caption{(a)~A graph $\Gr$ with a port numbering and an orientation. (b)~A proper labelling $\ell_{\Gr}$ that is derived from the port numbering. We have an $L$-digraph with $L = \{ a,b,c \}$, $a = (1,2)$, $b = (2,1)$, and $c = (3,1)$. (c)~The view of $\Gr$ from $u$ is an infinite directed tree $\Tr = \Tr(\Gr,u)$; there is a covering map $\varphi$ from $\Tr$ to $\Gr$ that preserves adjacencies, orientations, and edge labels. For example, $\varphi(\lambda) = \varphi(aab^{-1}) = u$.}\label{fig:ldigraph}
\end{figure}

\paragraph{Views.}

The information available to a $\PO$-algorithm computing on a node $u\in V(\Gr)$ in an $L$-digraph $\Gr$ is usually modelled as follows~\cite{angluin80local, yamashita96computing, suomela09survey}. The \emph{view} of $\Gr$ from $u$ is an $L$-edge-labelled rooted (possibly infinite) directed tree $\Tr=\Tr({\Gr},u)$, where the vertices $V(\Tr)$ correspond to all non-backtracking walks on $\Gr$ starting at $u$; see Figure~\ref{fig:ldigraph}c. Formally, a $k$-step walk can be identified with a word of length $k$ in the letters $L\cup L^{-1}$. A non-backtracking walk is a \emph{reduced} word where neither $\ell\ell^{-1}$ nor $\ell^{-1}\ell$ appear. If $w\in V(\Tr)$ is a walk on $\Gr$ from $u$ to $v$, we define $\varphi(w) = v$. In particular, the root of $\Tr$ is the \emph{empty word} $\lambda$ with $\varphi(\lambda) = u$. The directed edges of $\Tr$ (and their labels) are defined in such a way that $\varphi\colon V(\Tr) \to V(\Gr)$ becomes a covering map. Namely, $w\in V(\Tr)$ has an out-neighbour $w\ell$ for every $\ell\in L$ such that $\varphi(w)$ has a outgoing edge labelled $\ell$.

\paragraph{Local $\boldsymbol\PO$-Algorithms.}

The inability of a $\PO$-algorithm $\Blg$ to detect cycles in a graph is characterised by the fact that $\Blg(\Gr,u) = \Blg(\Tr(\Gr,u))$. In fact, we \emph{define} a local $\PO$-algorithm as a function $\Blg$ satisfying $\Blg(\Gr,u)=\Blg(\tau(\Tr(\Gr,u)))$. 
An important consequence of this definition is that the output of a $\PO$-algorithm is invariant under lifts, i.e., if $\varphi\colon V(\Hr)\to V(\Gr)$ is a covering map of $L$-digraphs, then $\Blg(\Hr,u) = \Blg(\Gr,\varphi(u))$. The intuition is that nodes in a common fibre are always in the same state during computation as they see the same view.

The following formalism will become useful. Denote by $(\Tr^*,\lambda)$ the complete $L$-labelled rooted directed tree of radius $r$ with $V(\Tr^*)$ consisting of reduced words in the letters $L\cup L^{-1}$, i.e., every non-leaf vertex in $\Tr^*$ has an outgoing edge and an incoming edge for each $\ell\in L$; see Figure~\ref{fig:complete}. The output of $\Blg$ on every graph $\Gr\in\Ff$ is completely determined after specifying its output on the subtrees of $(\Tr^*,\lambda)$. More precisely, let $\Wf$ consist of vertex sets $W\subseteq V(\Tr^*)$ such that $(\Tr^*,\lambda)\upharpoonright W = \tau(\Tr(\Gr,u))$ for some $\Gr\in\Ff$ and $u\in V(\Gr)$. Then a function $\Blg\colon\Wf \to \Omega$ defines a $\PO$-algorithm by identifying $\Blg((\Tr^*,\lambda)\upharpoonright W) = \Blg(W)$.

\begin{figure}
    \centering
    \includegraphics[page=\PComplete]{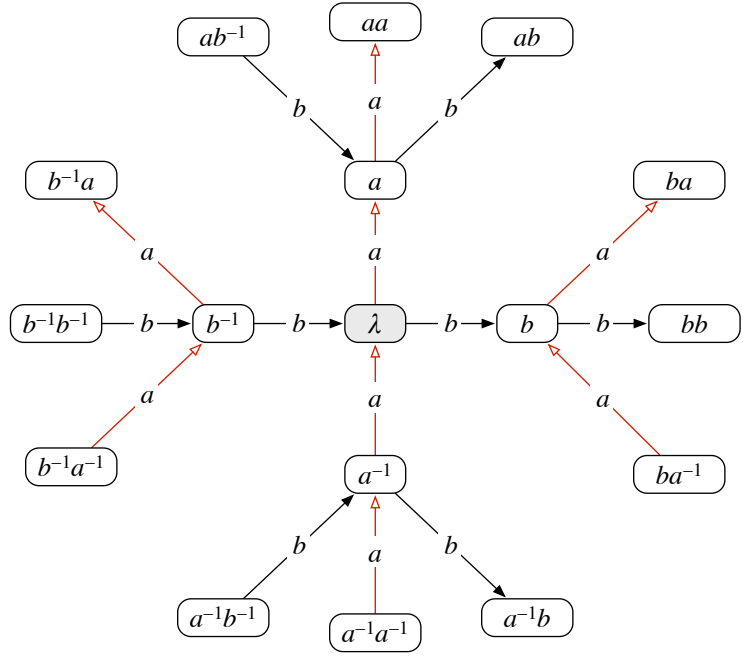}
    \caption{The complete $L$-labelled rooted directed tree $(\Tr^*,\lambda)$ of radius $r = 2$, for $L = \{a,b\}$.}\label{fig:complete}
\end{figure}

\section{Order Homogeneity}

In this section we introduce some key concepts that are used in controlling the local symmetry breaking information that is available to a local $\OI$-algorithm.

\subsection{Homogeneous Graphs}\label{ssec:homog}

In the following, we take the \emph{isomorphism type} of an $r$-neighbourhood $\tau=\tau(\Gr,\<,u)$ to be some canonical representative of the isomorphism class of $\tau$.
\begin{definition}
    Let $(\Hr,\<)$ be an ordered graph. If there is a set $U\subseteq V(\Hr)$ of size $|U| \ge \alpha|\Hr|$ such that the vertices in $U$ have a common $r$-neighbourhood isomorphism type $\tau^*$, then we call $(\Hr,\<)$ an \emph{$(\alpha,r)$-homogeneous graph} and $\tau^*$ the associated \emph{homogeneity type} of $\Hr$.
\end{definition}

Homogeneous graphs are useful in fooling $\OI$-algorithms: an $(\alpha,r)$\hyp homogeneous graph forces any local $\OI$-algorithm to produce the same output in at least an $\alpha$ fraction of the nodes in the input graph. However, there are some limitations to how large $\alpha$ can be: Let $(\Gr,\<)$ be a connected ordered graph on at least two vertices. If $u$ and $v$ are the smallest and the largest vertices of $\Gr$, their $r$-neighbourhoods $\tau(\Gr,\<,u)$ and $\tau(\Gr,\<,v)$ cannot be isomorphic even for $r=1$. Thus, non-trivial finite graphs are not $(1,1)$-homogeneous. Moreover, an ordered $(2k-1)$-regular graph cannot be $(\alpha,1)$-homogeneous for any $\alpha > 1/2$; this is the essence of the weak $2$-colouring algorithm of Naor and Stockmeyer~\cite{naor95what}.

\begin{figure}
    \centering
    \includegraphics[page=\PHomogTree]{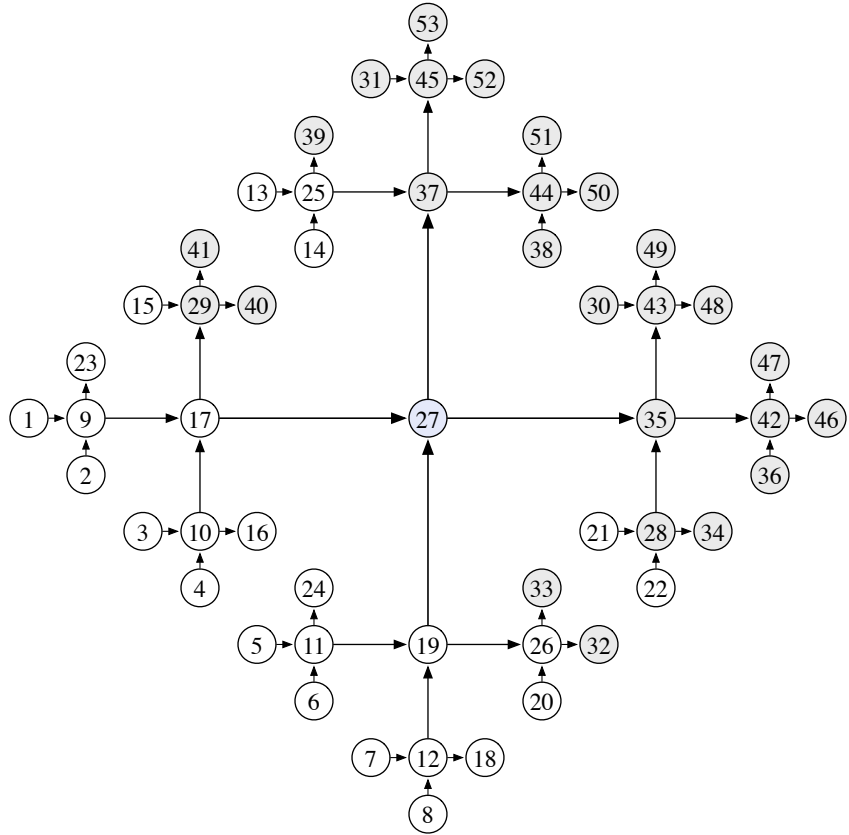}
    \caption{A fragment of a $4$-regular infinite ordered tree $(\Gr,\<)$. The numbering of the nodes indicates a $(1,r)$-homogeneous linear order in the neighbourhood of node $27$; grey nodes are larger than $27$ and white nodes are smaller than $27$.}\label{fig:homog-tree}
\end{figure}

\begin{figure}
    \centering
    \includegraphics[page=\PHomog]{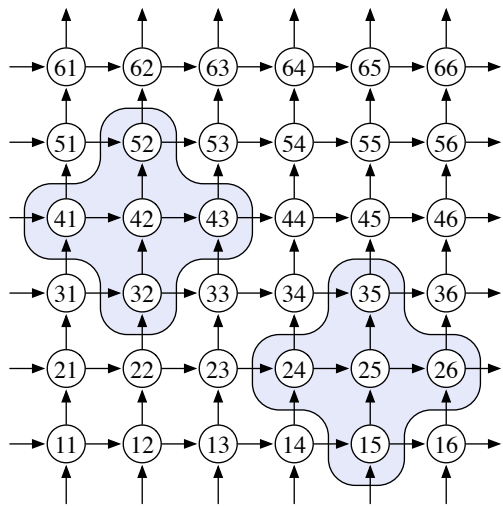}
    \caption{A $4$-regular graph $\Gr$ constructed as the cartesian product of two directed $6$-cycles. We define the ordered graph $(\Gr,\<)$ by choosing the linear order $11 < 12 < \dotsb < 16 < 21 < 22 < \dotsb < 66$. The radius-$1$ neighbourhood of node $25$ is isomorphic to the radius-$1$ neighbourhood of node $42$. In general, there are $16$ nodes (fraction $4/9$ of all nodes) that have isomorphic radius-$1$ neighbourhoods; hence $(\Gr,\<)$ is $(4/9, 1)$-homogeneous. It is also $(1/9, 2)$-homogeneous.}\label{fig:homog}
\end{figure}

Our main technical tool will be a construction of graphs that satisfy the following properties:
\begin{enumerate}[label=(\arabic*),noitemsep,align=left,labelwidth=4ex,leftmargin=8ex]
    \item $(1-\epsilon,r)$-homogeneous for any $\epsilon > 0$ and $r$,
    \item $2k$-regular for any $k$,
    \item large girth,
    \item finite order.
\end{enumerate}
Note that it is relatively easy to satisfy any three of these properties:
\begin{itemize}[align=left,labelwidth=13ex,leftmargin=17ex]
    \item[(1), (2), (3)] Infinite $2k$-regular trees admit a $(1,r)$-homogeneous linear order; see Figure~\ref{fig:homog-tree} for an example.
    \item[(1), (2), (4)] We can construct a sufficiently large $k$-dimensional toroidal grid graph (cartesian product of $k$ directed cycles) and order the nodes lexicographically coordinate-wise; see Figure~\ref{fig:homog} for an example. However, these graphs have girth $4$ when $k \ge 2$.
    \item[(1), (3), (4)] A sufficiently large directed cycle is $(1-\epsilon,r)$-homogeneous and has large girth. However, all the nodes have degree~$2$.
    \item[(2), (3), (4)] It is well known that regular graphs of arbitrarily high girth exist.
\end{itemize}
Our construction satisfies all four properties simultaneously.
\begin{thm}\label{thm:homog-graph}
Let $k,r\in\N$. For every $\epsilon > 0$ there exists a finite $2k$-regular $(1-\epsilon,r)$-homogeneous connected graph $(\Hr_\epsilon,\<_\epsilon)$ of girth larger than $2r+1$. Furthermore, the following properties hold:
\begin{enumerate}
	\item The homogeneity type $\tau^*$ of $(\Hr_\epsilon,\<_\epsilon)$ does not depend on $\epsilon$.
	\item The graph $\Hr_\epsilon$ and the type $\tau^*$ are $k$-edge-labelled digraphs.
\end{enumerate}
\end{thm}
We defer the proof of Theorem~\ref{thm:homog-graph} to Section~\ref{sec:homog-graphs}. There, it turns out that Cayley graphs of \emph{soluble} groups suit our needs: The homogeneous toroidal graphs mentioned above are Cayley graphs of the abelian groups $\Z_n^k$. Analogously, we use the decomposition of a soluble group into abelian factors to guarantee the presence of a suitable ordering. However, to ensure large girth, the groups we consider must be sufficiently far from being abelian, i.e., they must have large derived length~\cite{conder10limitations}.

\subsection{Homogeneous Lifts}\label{ssec:homog-lift}

We fix some notation towards a proof of Theorem~\ref{thm:main}. By Theorem~\ref{thm:homog-graph} we let $(\Hr_\epsilon,\<_\epsilon)$, $\epsilon > 0$, be a family of $2|L|$-regular $(1-\epsilon,r)$-homogeneous connected graphs of girth $>2r+1$ interpreted as $L$-digraphs. The homogeneity type $\tau^*$ that is shared by all $\Hr_\epsilon$ is then of the form $\tau^*=(\Tr^*,\<^*,\lambda)$, where $\Tr^*$ is the complete $L$-labelled tree of Section~\ref{sec:po-model}.

We use the graphs $\Hr_\epsilon$ to prove the following theorem.

\begin{thm}\label{thm:subtree}
Let $\Gr$ be an $L$-digraph. For every $\epsilon > 0$ there exists a lift $(\Gr_\epsilon,\<_{\Gr\epsilon})$ of $\Gr$ such that a $(1-\epsilon)$ fraction of the vertices in $(\Gr_\epsilon,\<_{\Gr\epsilon})$ have $r$-neighbourhoods isomorphic to a subtree of $\tau^*=(\Tr^*,\<^*,\lambda)$. Moreover, if $\Gr$ is connected, $\Gr_\epsilon$ can be made connected.
\end{thm}
\begin{proof}
Write $(\Cr,\<_{\Cr})=(\Gr_\epsilon,\<_{\Gr\epsilon})$ and $(\Hr,\<_{\Hr}) = (\Hr_\epsilon,\<_\epsilon)$ for short.
Our goal is to construct $(\Cr,\<_{\Cr})$ as a certain product of $(\Hr,\<_{\Hr})$ and $\Gr$; see Figure~\ref{fig:product}. This product is a modification of the common lift construction of Angluin and Gardiner~\cite{angluin81finite}.

\begin{figure}
    \centering
    \includegraphics[page=\PProduct]{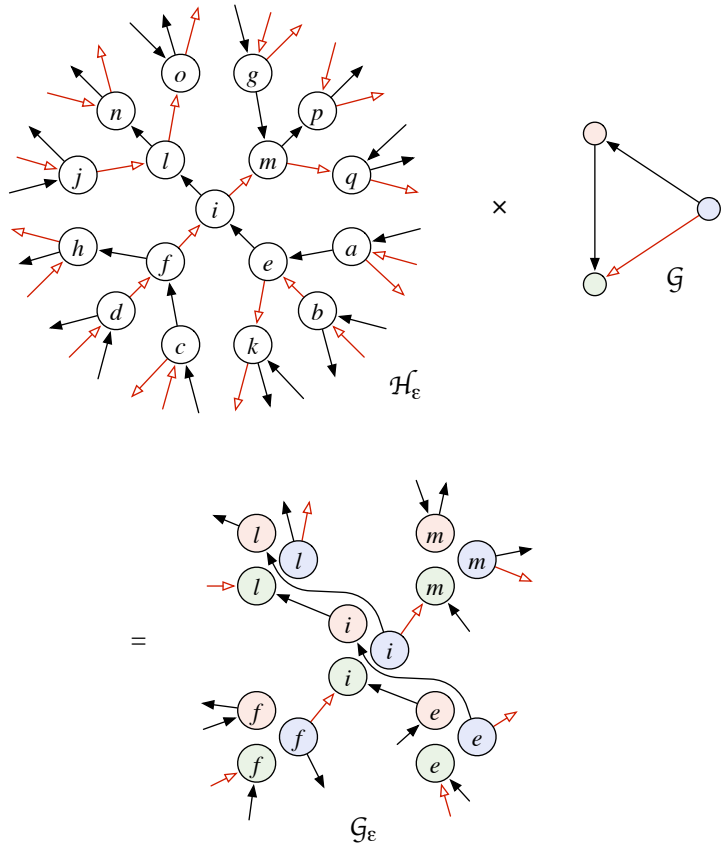}
    \caption{Homogeneous lifts. In this example $L = |2|$, and the two labels are indicated with two different kinds of arrows. Graph $\Hr_\epsilon$ is a homogeneous $2|L|$-regular ordered $L$-digraph with a large girth---in particular, the local neighbourhood of a node looks like a tree. Graph $\Gr$ is an arbitrary $L$-digraph, not necessarily ordered. Their product $\Gr_\epsilon$ is a lift of $\Gr$, but it inherits the desirable properties of $\Hr_\epsilon$: a high girth and a homogeneous linear order.}\label{fig:product}
\end{figure}

The lift $\Cr$ is defined on the product set $V(\Cr) = V(\Hr)\times V(\Gr)$ by ``matching equi-labelled edges'': the out-neighbours of $(h,g)\in V(\Cr)$ are vertices $(h',g')\in V(\Cr)$ such that $(h,h')\in E(\Hr)$, $(g,g')\in E(\Gr)$ and $\ell_\Hr(h,h')=\ell_\Gr(g,g')$. An edge $((h,g),(h',g'))\in E(\Gr)$ inherits the common label $\ell_\Hr(h,h')=\ell_\Gr(g,g')$.

The properties of $\Cr$ are related to the properties of $\Gr$ and $\Hr$ as follows.
\begin{enumerate}
    \item The projection $\varphi_{\Gr}\colon V(\Cr)\to V(\Gr)$ mapping $(h,g)\mapsto g$ is a covering map. This follows from the fact that each edge incident to $g\in V(\Gr)$ is always matched against an edge of $\Hr$ in the fibre $V(\Hr)\times\{g\}$.
    \item The projection $\varphi_{\Hr}\colon V(\Cr)\to V(\Hr)$ mapping $(h,g)\mapsto h$ is not a covering map in case $\Gr$ is not $2|L|$-regular. In any case $\varphi_{\Hr}$ is a graph homomorphism, and this implies that $\Cr$ has girth $> 2r+1$.
\end{enumerate}

Next, we define a partial order $<_p$ on $V(\Cr)$ as $u<_p v \iff \varphi_{\Hr}(u) <_{\Hr} \varphi_{\Hr}(v)$, for $u,v\in V(\Cr)$. Note that this definition leaves only pairs of vertices in a common $\varphi_{\Hr}$-fibre incomparable. But since $\Hr$ has large girth, none of the incomparable pairs appear in an $r$-neighbourhood of $\Cr$. We let $<_{\Cr}$ be any completion of $<_p$ into a linear order. The previous discussion implies that $<_{\Cr}$ satisfies $\tau(\Cr,\<_{\Cr},u)=\tau(\Cr,\<_p,u)$ for all $u\in V(\Cr)$.

Let $U_{\Hr} \subseteq V(\Hr)$, $|U_{\Hr}| \ge (1-\epsilon)|\Hr|$, be the set of type $\tau^*$ vertices in $(\Hr,\<_{\Hr})$. Set $U_{\Cr} = \varphi^{-1}_{\Hr}(U_{\Hr})$ so that $|U_{\Cr}| \ge (1-\epsilon)|\Cr|$. Let $u\in U_{\Cr}$. By our definition of $<_p$, $\varphi_{\Hr}$ maps the $r$-neighbourhood $\tau_u=\tau(\Cr,\<_{\Cr},u)$ into $\tau(\Hr,\<_{\Hr},\varphi_{\Hr}(u))\simeq \tau^*$ while preserving the order. But because $\tau^*$ is a tree, $\varphi_{\Hr}$ must be injective on the vertex set of $\tau_u$ so that $\tau_u$ is isomorphic to a subtree of $\tau^*$ as required.

Finally, suppose $\Gr$ is connected. Then, by averaging, some connected component of $\Cr$ will have vertices in $U_{\Cr}$ with density at least $(1-\epsilon)$. This component satisfies the theorem.
\end{proof}

\section{Proof of Main Theorem}\label{sec:proof-mainthm}

Next, we use the tools of the previous section to prove Theorem~\ref{thm:main}. For clarity of exposition we first prove Theorem~\ref{thm:main} in the special case where $\Alg$ is an $\OI$-algorithm. The subsequent proof for an $\ID$-algorithm $\Alg$ uses a somewhat technical but well-known Ramsey type argument.

\subsection{Proof of Main Theorem for \texorpdfstring{$\boldsymbol\OI$}{OI}-algorithms}\label{ssec:mainthm-oi}

We will prove the general and connected versions of Theorem~\ref{thm:main} simultaneously; for the proof of the connected version it suffices to consider only connected lifts below. We do not need the assumption that $\Ff$ does not contain any trees.

Let $\PI$ be as in the statement of Theorem~\ref{thm:main}. Suppose an $\OI$-algorithm $\Alg$ finds an \Apx{\alpha} of $\PI$ in $\Ff$. We define a $\PO$-algorithm $\Blg$ simply by setting for $W \in \Wf$,
\[
    \Blg(W) = \Alg\bigl((\Tr^*,\<^*,\lambda)\upharpoonright W\bigr).
\]
Now, Theorem~\ref{thm:subtree} translates into saying that for every $\Gr\in\Ff$ and $\epsilon > 0$ we have that $\Alg(\Gr_\epsilon,\<_{\Gr\epsilon},u)=\Blg(\Gr_\epsilon,u)$ for at least a $(1-\epsilon)$ fraction of nodes $u\in V(\Gr_\epsilon)$. The claim that $\Blg$ works as expected follows essentially from this fact as we argue next.

For simplicity, we assume the solutions to $\PI$ are sets of vertices so that $\Alg(\Gr)\subseteq V(\Gr)$; solutions that are sets of edges are handled similarly.

Fix $\Gr\in\Ff$ and let $\varphi_\epsilon\colon V(\Gr_\epsilon)\to V(\Gr)$, $\epsilon > 0$, be the associated covering maps.

\paragraph{Algorithm $\boldsymbol\Blg$ Finds a Feasible Solution of $\boldsymbol\PI$ on $\boldsymbol\Gr$.}

Let $\Vlg$ be a local $\PO$-algorithm verifying the feasibility of a solution for $\PI$; we may assume $\Vlg$ also runs in time $r$. For $\epsilon>0$ sufficiently small, each $v\in V(\Gr)$ has a pre-image $v'\in \varphi_\epsilon^{-1}(v)$ such that $\Alg$ and $\Blg$ agree on the vertices $\bigcup_{v\in V(\Gr)} B_{\Gr\epsilon}(v',r)$. Thus, $\Vlg$ accepts the solution $\Blg(\Gr_\epsilon)$ on the vertices $v'$. But because $\varphi_\epsilon(\{v':v\in V(\Gr)\}) = V(\Gr)$ it follows that $\Vlg$ accepts the solution $\Blg(\Gr) = \varphi_\epsilon(\Blg(\Gr_\epsilon))$ on every node in $\Gr$.

\paragraph{Algorithm $\boldsymbol\Blg$ Finds an $\boldsymbol\alpha$-Approximation of $\boldsymbol\PI$ on $\boldsymbol\Gr$.}

We assume $\PI$ is a minimisation problem; maximisation problems are handled similarly. Let $X\subseteq V(\Gr)$ and $X_\epsilon\subseteq V(\Gr_\epsilon)$ be some optimal solutions of $\PI$.

As $\epsilon \to 0$, the solutions $\Blg(\Gr_\epsilon)$ and $\Alg(\Gr_\epsilon)$ agree on almost all the vertices. Indeed, a simple calculation shows that $|\Blg(\Gr_\epsilon)| \le f(\epsilon)\cdot|\Alg(\Gr_\epsilon)|$ for some $f$ with $f(\epsilon)\to 1$ as $\epsilon \to 0$. Furthermore,
\[
    \frac{|\Blg(\Gr)|}{|X|}
    = \frac{|\varphi_\epsilon^{-1}(\Blg(\Gr))|}{|\varphi_\epsilon^{-1}(X)|}
    \le \frac{|\Blg(\Gr_\epsilon)|}{|X_\epsilon|}
    \le \frac{f(\epsilon)\cdot|\Alg(\Gr_\epsilon)|}{|X_\epsilon|}
    \le f(\epsilon)\alpha,
\]
where the first equality follows from $\varphi_\epsilon$ being an $n$-lift, and
the first inequality follows from $\varphi^{-1}_\epsilon(\Blg(\Gr))=\Blg(\Gr_\epsilon)$ and the fact that $\varphi^{-1}_\epsilon(X)$ is a feasible solution so that $|X_\epsilon| \le |\varphi^{-1}_\epsilon(X)|$. Since the above inequality holds for every $\epsilon > 0$ we must have that $|\Blg(\Gr)|/|X| \le \alpha$, as desired.

\subsection{Proof of Main Theorem for \texorpdfstring{$\boldsymbol\ID$}{ID}-algorithms}\label{ssec:mainthm-id}

We extend the above proof to the case of local $\ID$-algorithms $\Alg$ by designing ``worst-case'' vertex identifiers for the instances in $\Ff$ in order to make $\Alg$ behave similarly to a $\PO$-algorithm on tree neighbourhoods. To do this we use the Ramsey technique of Naor and Stockmeyer~\cite{naor95what}; see also Czygrinow et al.~\cite{czygrinow08fast}. For a reference on Ramsey's theorem see Graham et al.~\cite{graham80ramsey}.

We use the following notation: if $(X,\<_X)$ and $(Y,\<_Y)$ are linearly ordered sets with $|X| \le |Y|$, we write $f\colon (X,\<_X)\hookrightarrow (Y,\<_Y)$ for the unique order-preserving injection $f\colon X\to Y$ that maps the $i$th element of $X$ to the $i$th element of $Y$. A \emph{$t$-set} is a set of size $t$, and the set of $t$-subsets of $X$ is denoted $X^{(t)}$.

Write $\Omega^\Wf$ for the family of functions $\Wf\to\Omega$; recall that each $\Blg\in \Omega^\Wf$ can be interpreted as a $\PO$-algorithm. Set $k=|\Omega^\Wf|$ and $t=|\Tr^*|$. We consider every $t$-subset $A\in\N^{(t)}$ to be ordered by the usual order $<$ on $\N$. For $W\in\Wf$ we let $f_{W,A}\colon (W,\<^*)\hookrightarrow(A,\<)$ so that the vertex-relabelled tree $f_{W,A}((\Tr^*,\lambda)\upharpoonright W)$ has the $|W|$ smallest numbers in $A$ as vertices. Define a $k$-colouring $c\colon \N^{(t)}\to\Omega^\Wf$ by setting
\[
    c(A)(W) = \Alg(f_{W,A}((\Tr^*,\lambda)\upharpoonright W)). 
\]

For each $m \ge t$ we can use Ramsey's theorem to obtain a number $R(m) \ge m$, so that for every $R(m)$-set $I\subseteq\N$ there exists an $m$-subset $J\subseteq I$ such that $J^{(t)}$ is monochromatic under $c$, i.e., all $t$-subsets of $J$ have the same colour. In particular, for every interval
\[
    I(m,i)=[(i-1) R(m)+1,\, i R(m)], \quad i \ge 1,
\]    
there exist an $m$-subset $J(m,i)\subseteq I(m,i)$ and a colour (i.e., an algorithm) $\Blg_{m,i} \in \Omega^\Wf$ such that $c(A) = \Blg_{m,i}$ for all $t$-subsets $A \subseteq J(m,i)$.

This construction has the following property.
\begin{prop}\label{prop:agreement}
Suppose $m \ge |\Gr_\epsilon|+t$. Algorithms $\Alg$ and $\Blg_{m,i}$ produce the same output on at least a  $(1-\epsilon)$ fraction of the vertices in the vertex-relabelled $L$-digraph $f_{m,i}(\Gr_\epsilon)$, where 
\[
    f_{m,i}\colon (V(\Gr_\epsilon),\<_{\Gr\epsilon})\hookrightarrow(J(m,i),\<).
\]
\end{prop}
\begin{proof}
By Theorem~\ref{thm:subtree}, let $U\subseteq V(f_{m,i}(\Gr_\epsilon))$, $|U| \ge (1-\epsilon)|\Gr_\epsilon|$, be the set of vertices $u$ with $\tau(f_{m,i}(\Gr_\epsilon),\<,u)$ isomorphic to a subtree of $\tau^*$. In particular, for a fixed  $u\in U$ we can choose $W\in \Wf$ such that
\[
    \tau(f_{m,i}(\Gr_\epsilon),\<,u) \simeq (\Tr^*,\<^*,\lambda)\upharpoonright W.
\]
Now, as $m$ is large, there exists a $t$-set $A\subseteq J(m,i)$ such that
\[
    \tau(f_{m,i}(\Gr_\epsilon),u) = f_{W,A}((\Tr^*,\lambda)\upharpoonright W).
\]
Thus, $\Alg$ and $\Blg_{m,i}$ agree on $u$ by the definition of $\Blg_{m,i}$.
\end{proof}

For every $n\in\N$ some colour appears with density at least $1/k$ (i.e., appears at least $n/k$ times) in the sequence $\Blg_{m,1},\Blg_{m,2},\dotsc,\Blg_{m,n}$. Hence, let $\Blg_m$ be a colour that appears with density at least $1/k$ among these sequences for infinitely many $n$. Let $\Blg$ be a colour appearing among the $\Blg_m$ for infinitely many $m$. We claim $\Blg$ satisfies Theorem~\ref{thm:main}. In fact, Theorem~\ref{thm:main} follows from the following proposition together with the considerations of Section~\ref{ssec:mainthm-oi}.
\begin{prop}
For every $\Gr_\epsilon$ there exists an $n$-lift $\Hr$ of $\Gr_\epsilon$ such that $V(\Hr)\subseteq\{1,2,\dotsc,s(|\Hr|)\}$ and $\Alg(\Hr,u) = \Blg(\Hr,u)$ for a $(1-\epsilon)$ fraction of nodes $u\in V(\Hr)$. Moreover, if $\Gr_\epsilon$ is connected and not a tree, $\Hr$ can be made connected.
\end{prop}
\begin{proof}
Let $m$ be such that $m \ge |\Gr_\epsilon|+t$ and $\Blg = \Blg_m$. For infinitely many $n$ there exists an $n$-set $I\subseteq[nk]$ of indices such that $\Blg = \Blg_{m,i}$ for $i\in I$. Consider the following $n$-lift of $\Gr_\epsilon$ obtained by taking disjoint unions:
\[
    \Hr = \bigcup_{i\in I} f_{m,i}(\Gr_\epsilon).
\]
Algorithms $\Alg$ and $\Blg$ agree on a $(1-\epsilon)$ fraction of the nodes in $\Hr$ by Proposition~\ref{prop:agreement}. Furthermore, we have $|\Hr|=n|\Gr_\epsilon|$ and $V(\Hr)\subseteq\{1,2,\dotsc,n k R(m)\}$. We are assuming that $s(n)=\omega(n)$ so choosing a large enough $n$ proves the non-connected version of the claim.

Finally, suppose $\Gr_\epsilon$ is connected and not a tree. We may assume that there is an edge $e=(u,v)\in E(\Gr_\epsilon)$ so that $\Gr_\epsilon$ remains connected when $e$ is removed and that a $(1-\epsilon)$ fraction of vertices in $\Gr_\epsilon$ have $r$-neighbourhoods not containing $e$ that are isomorphic into $\tau^*$. Now $\Hr$ above is easily modified into a connected graph by redefining the directed matching between the fibre $\{u_i\}_{i\in I}$ of $u$ and the fibre $\{v_i\}_{i\in I}$ of $v$. Namely, let $\pi$ be a cyclic permutation on $I$ and set
\[
    E' = \bigl(E(\Hr) \smallsetminus \{(u_i,v_i)\}_{i\in I} \bigr) \,\cup\, \{(u_i,v_{\pi(i)})\}_{i\in I}.
\]
Then $\Hr'=(V(\Hr), E')$ is easily seen to be a connected $n$-lift of $\Gr_\epsilon$ satisfying the claim.
\end{proof}

\begin{remark}\label{rem:identifiers}
Above, we assumed that instances $\Gr$ have node identifiers $V(\Gr)\subseteq \{1,2,\dotsc,s(n)\}$, $n=|\Gr|$, for $s(n)=\omega(n)$. By choosing identifiers more economically as in the work of Czygrinow et al.~\cite{czygrinow08fast} one can show lower bounds for the graph problems of Section~\ref{ssec:intro-local-apx} even when $s(n)=n$.
\end{remark}

\section{Construction of Homogeneous Graphs of Large Girth}\label{sec:homog-graphs}

In this section we prove Theorem~\ref{thm:homog-graph}. Our construction uses Cayley graphs of semi-direct products of groups. First, we recall the terminology in use here; for a standard reference on group theory see, e.g., Rotman~\cite{rotman95introduction}.

For the benefit of the reader who is not well-versed in group theory we include in Appendix~\ref{app:wreath} a short primer on the semi-direct product groups that are used below.

\subsection{Semi-Direct Products}

Let $G$ and $H$ be groups with $H$ acting on $G$ as a group of automorphisms. We write $h\cdot g$ for the action of $h\in H$ on $g\in G$ so that the mapping $g\mapsto h\cdot g$ is an automorphism of $G$. The \emph{semi-direct product} $G\rtimes H$ is defined to be the set $G\times H$ with the group operation given by
\[
    (g,h)(g',h') = (g(h\cdot g'),hh').
\]

\subsection{Cayley Graphs}

The \emph{Cayley graph} $\Cay(G,S)$ of a group $G$ with respect to a finite set $S\subseteq G$ is an $S$-digraph on the vertex set $G$ such that each $g\in G$ has an outgoing edge $(g,gs)$ labelled $s$ for each $s\in S$. We require that $1\notin S$ so as not to have any self-loops. We do not require that $S$ is a generating set for $G$, i.e., the graph $\Cay(G,S)$ need not be connected. 

If $\varphi\colon H\to G$ is an onto group homomorphism and $S\subseteq H$ is a set such that the mapping $\varphi$ is injective on $S\cup\{1\}$, then $\varphi$ naturally induces a covering map of digraphs $\Cay(H,S)$ and $\Cay(G,\varphi(S))$.

\subsection{Proof of Theorem~\texorpdfstring{\ref{thm:homog-graph}}{3}}\label{ssec:proof-homog-graph}

Let $n\in\N$ be an even number. We consider three families of groups, $\{H_i\}_{i \ge 1}$, $\{W_i\}_{i \ge 1}$, and $\{U_i\}_{i \ge 1}$, that are variations on a common theme. The families are defined iteratively as follows:
\begin{align*}
H_1 &= \Z_n, &
W_1 &= \Z_2, &
U_1 &= \Z, \\
H_{i+1} &= H_i^2 \rtimes \Z_n, &
W_{i+1} &= W_i^2 \rtimes \Z_2, &
U_{i+1} &= U_i^2 \rtimes \Z.
\end{align*}
Here, the cyclic group $\Z_n=\{0,1,\dotsc,n-1\}$ acts on the direct product $H_i^2 = H_i\times H_i$ by cyclically permuting the coordinates, i.e., the subgroup $2\Z_n \le \Z_n$ acts trivially and the elements in $1+ 2\Z_n$ swap the two coordinates. The groups $\Z_2$ and $\Z$ act analogously in the definitions of $W_i$ and $U_i$. See Appendix~\ref{app:wreath} for more information on groups $H_i$, $W_i$, and $U_i$.

The underlying sets of the groups $H_i$, $W_i$, and $U_i$ consist of $d(i)$-tuples of elements in $\Z$, for $d(i)=2^i-1$, so that $W_i\subseteq H_i\subseteq U_i$ \emph{as sets}. Interpreting these tuples as points in $\R^{d(i)}$ we immediately get a natural embedding of every Cayley graph of these groups in $\R^{d(i)}$. This geometric intuition will become useful later.

\begin{enumerate}
    \item The groups $W_i$ are $i$-fold iterated regular wreath products of the cyclic group $\Z_2$. These groups have order $|W_i|=2^{d(i)}$ and they are sometimes called \emph{symmetric $2$-groups}; they are isomorphic to the Sylow $2$-subgroups of the symmetric group on $2^i$ letters~\cite[p.\ 176]{rotman95introduction}.
    \item The groups $U_i$ are natural extensions of the groups $W_i$ by the free abelian group of rank $d(i)$: the mapping $\varphi_i\colon U_i\to W_i$ that reduces each coordinate modulo $2$ is easily seen to be an onto homomorphism with abelian kernel $(2\Z)^{d(i)} \simeq \Z^{d(i)}$.
    \item The groups $H_i$ are intermediate between $U_i$ and $W_i$ in that the mapping $\psi_i\colon U_i\to H_i$ that reduces each coordinate modulo $n$ is an onto homomorphism, and the mapping $\varphi_i'\colon H_i\to W_i$ that
reduces each coordinate modulo $2$ is an onto homomorphism. In summary, the following diagram commutes:
    \[
        \xymatrix{
            U_i \ar[r]^{\psi_i} \ar[rd]_{\varphi_i} &
            H_i \ar[d]^{\varphi_i'} \\
            & W_i
        }
    \]
\end{enumerate}

Our goal will be to construct a suitable Cayley graph $\Hr$ of some $H_i$. We will use the groups $W_i$ to ensure $\Hr$ has large girth, whereas the groups $U_i$ will guarantee that $\Hr$ has an almost-everywhere homogeneous linear ordering.

\paragraph{Girth.}

Gamburd et al.~\cite{gamburd09girth} study the girth of random Cayley graphs and prove, in particular, that a random $k$-subset of $W_i$ generates a Cayley graph of large girth with high probability when $i\gg k$ is large. We only need the following weaker version of their theorem (see Appendix~\ref{app:high-girth} for an alternative, constructive proof).

\begin{thm}[{Corollary to~\cite[Theorem~6]{gamburd09girth}}]\label{thm:high-girth} Let $k,r\in \N$.
There exists an $i\in\N$ and a set $S\subseteq W_i$, $|S|=k$, such that the girth of the Cayley graph $\Cay(W_i,S)$ is larger than $2r+1$.
\end{thm}

Fix a large enough $j\in\N$ and a $k$-set $S\subseteq W_j$ so that $\Cay(W_j,S)$ has a girth larger than $2r+1$. Henceforth, we omit the subscript $j$ and write $H$, $W$, $U$, $\varphi$, $\psi$ and $d$ in place of $H_{j}$, $W_{j}$, $U_{j}$, $\varphi_{j}$, $\psi_{j}$ and $d(j)$. Interpreting $S$ as a set of elements of $H$ and $U$ (so that $\varphi(S)=\psi(S)=S$) we construct the Cayley graphs
\[
    \Hr=\Cay(H,S), \quad
    \Wr = \Cay(W,S), \quad\text{and}\quad
    \Ur = \Cay(U,S). 
\]
As each of these graphs is a lift of $\Wr$, none have cycles of length at most $2r+1$ and their $r$-neighbourhoods are trees.

\paragraph{Linear Order.}

Next, we introduce a \emph{left-invariant} linear order $<$ on $U$ satisfying
\[
u< v \implies wu < wv,  \qquad \text{for all } u,v,w\in U.
\]
Such a relation can be defined by specifying a \emph{positive cone} $P\subseteq U$ of elements that are greater than the identity $1=1_{U}$ so that
\[
    u < v \iff 1 < u^{-1}v \iff u^{-1}v \in P.
\]
A relation $<$ defined this way is automatically left-invariant; it is transitive iff $u,v\in P$ implies $uv\in P$; and every pair $u\neq v$ is comparable iff for all $w\neq 1$, either $w\in P$ or $w^{-1}\in P$. The existence of a $P$ satisfying these conditions follows from the fact that $U$ is a torsion-free soluble group (e.g.,~\cite{conrad59right}), but it is easy enough to verify that setting 
\begin{equation}\label{eq:pos-cone}
    P = \bigl\{ (u_1,u_2,\dotsc,u_i, 0, 0, \dotsc, 0) \in U : 1 \le i \le d \text{ and } u_i > 0 \bigr\}
\end{equation}
satisfies the required conditions above (see Appendix~\ref{app:pos-cone}).

Because $U$ acts (by multiplication on the left) on $\Ur$ as a vertex-transitive group of graph automorphisms, it follows that the structures $(\Ur,\<,u)$, $u\in U$, are pairwise isomorphic. A fortiori, the $r$-neighbourhoods $\tau(\Ur,\<,u)$, $u\in U$, are all pairwise isomorphic. Let $\tau^*$ be this common $r$-neighbourhood isomorphism type.

\paragraph{Transferring the Linear Order on \texorpdfstring{$\boldsymbol U$}{U} to \texorpdfstring{$\boldsymbol\Hr$}{H}.}

Let $V(\Hr)$ be ordered by restricting the order $<$ on $U$ to the set $V(\Hr)=\Z_n^d$ underlying the group $H$. Note that $<$ is not a left-invariant order on $H$ (indeed, no non-trivial finite group can be left-invariantly ordered). Nevertheless, we will argue that, as $n\to\infty$, almost all $u\in V(\Hr)$ have $r$-neighbourhoods of type $\tau^*$.

The neighbours of a vertex $u\in V(\Ur)$ are elements $us$ where $s\in S\cup S^{-1}\subseteq [-1,1]^d$. The right multiplication action of $s\in S\cup S^{-1}$ on $u$ can be described in two steps as follows: First, the coordinates of $s$ are permuted (as determined by $u$) to obtain a vector $s'$. Then, $us$ is given as the standard addition of the vectors $u$ and $s'$ in $\Z^d\subseteq\R^d$. Hence, $us \in u+[-1,1]^d$, and moreover,
\begin{equation}\label{eq:bur}
    B_{\Ur}(u,r) \subseteq u+[-r,r]^d.
\end{equation}
This means that vertices close to $u$ in the graph $\Ur$ are also close in the associated geometric $\R^d$-embedding. 

Consider the set of inner nodes $I=[r,(n-1)-r]^d$. Let $u\in I$. By \eqref{eq:bur}, the vertex set $B_{\Ur}(u,r)$ is contained in $\Z_n^d$. This implies that the cover map $\psi$ is the identity on $B_{\Ur}(u,r)$ and consequently the $r$-neighbourhood $\tau(\Hr,\<,u)$ \emph{contains} the ordered tree $\tau(\Ur,\<,u)\simeq\tau^*$. If $\tau(\Hr,\<,u)$ had any additional edges to those of $\tau(\Ur,\<,u)$, this would entail a cycle of length $\le 2r+1$ in $\Hr$, which is not possible. Thus, $\tau(\Hr,\<,u)\simeq \tau^*$. The density of elements in $\Hr$ having $r$-neighbourhood type $\tau^*$ is therefore at least $|I|/|\Hr| =(n-2r)^d/n^d \ge 1-\epsilon$, for large $n$.

Finally, to establish Theorem~\ref{thm:homog-graph} it remains to address $\Hr$'s connectedness. But if $\Hr$ is not connected, an averaging argument shows that some connected component must have the desired density of at least $(1-\epsilon)$ of type $\tau^*$ vertices.

\section*{Acknowledgements}

We thank Christoph Lenzen and Roger Wattenhofer for discussions. This work was supported in part by the Academy of Finland, Grants 132380 and 252018, the Research Funds of the University of Helsinki, and the Finnish Cultural Foundation.

\newpage
\appendix

\section{Groups of Section~\texorpdfstring{\ref{sec:homog-graphs}}{5} in More Detail}\label{app:wreath}

This appendix contains expository material on the structure and properties of groups $H_i$, $W_i$, and $U_i$ for the convenience of the reader who is not familiar with semi-direct products. The notation we introduce here will become useful in Appendix~\ref{app:high-girth}.

\subsection{Binary Trees}

We can interpret the elements of groups $H_i$, $W_i$, and $U_i$ as complete binary trees of height $i$: there are $i$ levels of \emph{internal nodes}, one level of \emph{leaf nodes}, and all internal nodes have two children. The number of internal nodes is $d(i) = 2^i - 1$.

Moreover, there is some data associated with the internal nodes: in $H_i$ the internal nodes are labelled with the elements of $\Z_n$, in $W_i$ they are elements of $\Z_2$, and in $U_i$ they are elements of $\Z$.

\subsection{Group \texorpdfstring{$\boldsymbol{U_i}$}{Ui}}

We will now focus on the case of group $U_i$; the other two cases are similar. Wherever reasonable, we will use the convention that $x, y, \dotsc \in \Z$, $A, B, \dotsc \in U_i$, and $\alpha, \beta, \dotsc \in U_j$ for some $j < i$.

\paragraph{Base Case.}

We will use the symbol $\circ$ to denote a binary tree of height $0$. Group $U_0$ is the trivial group $U_0 = \{\circ\}$.

\paragraph{Recursive Step.}

Now assume that we have defined group $U_{i-1}$; we proceed to define group $U_i$. Elements of group $U_i$ are triples
$(\alpha, \beta, x)$, where
$\alpha, \beta \in U_{i-1}$ and
$x \in \Z$.
Intuitively, $(\alpha, \beta, x)$ is a tree with the element $x$ as the root node, $\alpha$ as the left subtree, and $\beta$ as the right subtree.

In what follows, we will write $\even{x}$ for an even $\even{x} \in \Z$ and $\odd{x}$ for an odd $\odd{x} \in \Z$. Group $\Z$ acts on $U_{i-1} \times U_{i-1}$ as follows, depending on the parity:
\begin{align*}
    \even{x} \cdot (\alpha, \beta) &= (\alpha, \beta), \\
    \odd{x} \cdot (\alpha, \beta) &= (\beta, \alpha).
\end{align*}
Hence by the definition of the semi-direct product, the group operation in $U_i$ is
\begin{align*}
    (\alpha, \beta, \even{x}) (\gamma, \delta, \even{y})
    &= ( \alpha\gamma,\ \beta\delta,\ \even{x} + \even{y}), \\
    (\alpha, \beta, \odd{x}) (\gamma, \delta, \even{y})
    &= ( \alpha\delta,\ \beta\gamma,\ \odd{x} + \even{y}), \\
    (\alpha, \beta, \even{x}) (\gamma, \delta, \odd{y})
    &= ( \alpha\gamma,\ \beta\delta,\ \even{x} + \odd{y}), \\
    (\alpha, \beta, \odd{x}) (\gamma, \delta, \odd{y})
    &= ( \alpha\delta,\ \beta\gamma,\ \odd{x} + \odd{y}).
\end{align*}

\paragraph{Shorthand Notation.}

To make this a bit easier to approach, let us define some shorthand notation:
\begin{align*}
    I_0 &= \circ, \\
    I_i &= (I_{i-1}, I_{i-1}, 0), \\
    X_i(x) &= (I_{i-1}, I_{i-1}, x), \\
    [\alpha, \beta] &= (\alpha, \beta, 0).
\end{align*}
In particular, $I_i \in U_i$ is an empty tree of height $i$: all internal nodes are labelled with zeroes. Observe that $I_i$ is the identity element of $U_i$:
\[
    (\alpha, \beta, x) I_i = (\alpha, \beta, x) = I_i (\alpha, \beta, x).
\]
Moreover, we can express any element as a product of $[\cdot,\cdot]$ and $X_i$:
\[
    [\alpha, \beta] \, X_i(x) = (\alpha, \beta, x).
\]
We will omit the subscript $i$ when it is clear from the context.

\paragraph{Examples of Group Operations.}

Now the group operations are much more straightforward:
\begin{align*}
    [\alpha, \beta] \, [\gamma, \delta] &= [\alpha\gamma, \beta\delta], \\
    X(x) \, X(y) &= X(x + y), \\
    X(\even{x}) \, [\alpha, \beta] &= [\alpha, \beta] \, X(\even{x}), \\
    X(\odd{x}) \, [\alpha, \beta] &= [\beta, \alpha] \, X(\odd{x}), \\
    [\alpha, \beta]^{-1} &= [\alpha^{-1}, \beta^{-1}], \\
    X(x)^{-1} &= X(-x).
\end{align*}
Here are some further examples:
\begin{align*}
    (\alpha, \beta, x)^{-1}
    &= \bigl([\alpha, \beta] \, X(x)\bigr)^{-1} \\
    &= X(x)^{-1} \, [\alpha, \beta]^{-1} \\
    &= X(-x) \, [\alpha^{-1}, \beta^{-1}], \\
    (\alpha, \beta, \even{x})^{-1}
    &= (\alpha^{-1}, \beta^{-1}, -\even{x}), \\
    (\alpha, \beta, \odd{x})^{-1}
    &= (\beta^{-1}, \alpha^{-1}, -\odd{x}).
\end{align*}
If $A \in U_i$, we can interpret the left multiplication as follows from the perspective of trees:
\begin{itemize}
    \item $X(\even{x}) \, A$: increment the label of the root node by $\even{x}$.
    \item $X(\odd{x}) \, A$: exchange the left and the right subtree of $A$, and increment the label of the root node by $\odd{x}$.
    \item $[\alpha,\beta] \, A$: recursively apply $\alpha$ to the left subtree of $A$ and $\beta$ to the right subtree of $A$.
\end{itemize}

\subsection{Positive Cone}\label{app:pos-cone}

Let us now have a closer look at the definition of a positive cone $P$ in \eqref{eq:pos-cone}. The properties of $P$ are easy to verify if we consider the following alternative, recursive definition. For each $i > 0$, we say that $A = (\alpha, \beta, x) \in U_i$ is \emph{positive} if one of the following holds:
\begin{enumerate}[noitemsep]
    \item $x > 0$,
    \item $x = 0$ and $\beta$ is positive,
    \item $x = 0$ and $\beta = I_{i-1}$ and $\alpha$ is positive.
\end{enumerate}
Finally, element $I_0 \in U_0$ is not positive. Let $P_i \subseteq U_i$ consist of all positive elements of $U_i$. Let us prove the following properties.

\begin{lemma}
    For any $i$ and $A \in U_i$, precisely one of the following is true: $A = I_i$, $A \in P_i$, or $A^{-1} \in P_i$.
\end{lemma}
\begin{proof}
    The proof is by induction. The case $i = 0$ is trivial: $U_0 = \{ I_0 \}$, and $I_0^{-1} = I_0$ is not positive by definition.
    
    Now assume that $i > 0$ and $A = (\alpha, \beta, x)$. If $x \ne 0$, we have $A^{-1} = (\gamma, \delta, -x)$ for some $\gamma, \delta \in U_{i-1}$. Therefore we have either (i)~$A \ne I_i$, $A \in P$, and $A^{-1} \notin P$, or (ii)~$A \ne I_i$, $A \notin P$, and $A^{-1} \in P$.
    
    Otherwise $x = 0$, and we have $A^{-1} = (\alpha^{-1}, \beta^{-1}, 0)$. The claim follows by a simple case analysis.
\end{proof}

\begin{lemma}
    If $A, B \in P_i$, we have $AB \in P_i$.
\end{lemma}
\begin{proof}
    The proof is by induction. The claim is trivial if $i = 0$. Otherwise, let $A = (\alpha_1, \beta_1, x)$ and $B = (\alpha_2, \beta_2, y)$. If $x > 0$ or $y > 0$, we have $AB = (\alpha, \beta, x+y) \in P_i$, as $x + y > 0$. Otherwise $x = y = 0$ and $AB = (\alpha_1 \alpha_2, \beta_1 \beta_2, 0)$, and the claim follows by a simple case analysis.
\end{proof}

Hence if we define a relation $<$ on $U_i$ by $A < B \iff A^{-1} B \in P_i$, we have the following properties:
\begin{enumerate}[noitemsep]
    \item $A < B$ and $B < C$ implies $A < B < C$,
    \item $A < B$ implies $CA < CB$, and
    \item for all $A, B \in U_i$ precisely one of the following is true: $A = B$, $A < B$, or $A > B$.
\end{enumerate}
That is, $<$ is a left-invariant linear order.

\subsection{Group \texorpdfstring{$\boldsymbol{W_i}$}{Wi}}

Let us now compare $W_i$ with $U_i$. The case of $W_i$ is further simplified, as we have only two possible values of $x \in \{0,1\}$. Moreover, $X(0) = I$ is the identity element; hence there is only one non-trivial case, $X(1)$:
\begin{align*}
    X(1) \, [\alpha, \beta] &= [\beta, \alpha] \, X(1), \\
    X(1)^{-1} &= X(1), \\
    (\alpha, \beta, 0)^{-1}
    &= (\alpha^{-1}, \beta^{-1}, 0), \\
    (\alpha, \beta, 1)^{-1}
    &= (\beta^{-1}, \alpha^{-1}, 1).
\end{align*}
Intuitively, an element $A \in W_i$ is simply an automorphism of a complete binary tree of height $i$: for example, $X(1)$ is an automorphism that exchanges the left and the right subtree of the root node. The group operation is, in essence, the composition of automorphisms.

\paragraph{More Shorthand Notation.}

The following notation will prove useful in Appendix~\ref{app:high-girth}. First, let us extend the bracket notation to cover an arbitrary vector of length $\ell = 2^k$:
\begin{align*}
    [ \alpha_1, \alpha_2, \alpha_3, \alpha_4 ] &= \bigl[ [\alpha_1, \alpha_2],\, [\alpha_3, \alpha_4] \bigr], \\
    [ \alpha_1, \alpha_2, \dotsc, \alpha_8 ] &= \Bigl[ \bigl[ [\alpha_1, \alpha_2],\, [\alpha_3, \alpha_4] \bigr],\ \bigl[ [\alpha_5, \alpha_6],\, [\alpha_7, \alpha_8] \bigr] \Bigr], \ \dotsc
\end{align*}
Then, we extend the function $X$ in an analogous fashion. For any $\ell = 2^k$, we define
\[
    X_i(x_1, x_2, \dotsc, x_\ell) = [ X_{i-k}(x_1),\, X_{i-k}(x_2),\, \dotsc,\, X_{i-k}(x_\ell) ].
\]
and
\[
    X_i(r; \ell) = X_i(y_1, y_2, \dotsc, y_\ell), \text{where } y_r = 1 \text{ and } y_j = 0 \text{ for all } j \ne r.
\]
For example,
\[
    X(3; 4) = X(0, 0, 1, 0) = \bigl[ [ X(0), X(0) ],\, [ X(1), X(0) ] \bigr].
\]
Intuitively, $X_i(x_1, x_2, \dotsc, x_\ell)$ is a complete binary tree of height $i$ that is mostly empty: the nodes at depth $k$ have values $x_1, x_2, \dotsc, x_\ell$, but all other nodes have value $0$.

\paragraph{Examples of Group Operations.}

We conclude this section with the following examples of group operations in $W_i$:
\begin{align*}
    X(x_1, x_2, \dotsc, x_\ell)^{-1}
        &= X(x_1, x_2, \dotsc, x_\ell), \\
    X(x_1, x_2, \dotsc, x_\ell) \, X(y_1, y_2, \dotsc, y_\ell)
        &= X(x_1 + y_1,\, x_2 + y_2,\, \dotsc,\, x_\ell + y_\ell), \\
    X(x_1, x_2, \dotsc, x_\ell) \, [ \alpha_1, \alpha_2, \dotsc, \alpha_\ell ]
        &= [ X(x_1)\,\alpha_1,\, X(x_2)\,\alpha_2,\, \dotsc,\, X(x_\ell)\,\alpha_\ell ].
\end{align*}
In particular, $X(r;\ell)$ exchanges the $r$th pair of subtrees:
\begin{align*}
    X(3;4) \, [\alpha_1,\beta_1,\alpha_2,\beta_2,\alpha_3,\beta_3,\alpha_4,\beta_4]
        &= [\alpha_1,\beta_1,\alpha_2,\beta_2,\beta_3,\alpha_3,\alpha_4,\beta_4] X(3;4).
\end{align*}

\section{Constructing High-Girth Generators}\label{app:high-girth}

In this appendix, we give an alternative proof of Theorem~\ref{thm:high-girth}. Our proof borrows many ideas from the original proof of Gamburd et al.~\cite{gamburd09girth}. However, while they use the probabilistic method to prove that a set $S$ exists, we give a simple explicit construction. In terms of the asymptotic growth of girth, this version is weaker but sufficient for our purposes. We will use the notation defined in Appendix~\ref{app:wreath}; see Figure~\ref{fig:high-girth} for an illustration.

\subsection{Preliminaries}

Fix integers $g$ and $m$. For $i = 0, 1, \dotsc, g$, define
\begin{align*}
    f(i) &= m+g-i, \\
    h(i) &= i (2m+2g-i+1) / 2, \\
    n(i) &= 2^{f(i)}.
\end{align*}
Note that we have
\begin{align*}
    h(0) &= 0, \\
    h(i) - h(i-1) &= f(i) + 1, \\
    n(i) / n(i-1) &= 1/2.
\end{align*}

For each $i = 1,2,\dotsc,g$, let $G_i$ be the direct product
\[
    G_i = W_{h(i-1)+1}^{n(i)} = W_{h(i-1)+1} \times W_{h(i-1)+1} \times \dotso \times W_{h(i-1)+1}.
\]
As $h(i-1) + 1 + \log_2 n(i) = h(i)$, we can interpret $G_i$ as a subgroup of $W_{h(i)}$: it consists of elements of the form
\[
    [ \alpha_1, \alpha_2, \dotsc, \alpha_{n(i)}] \in W_{h(i)}, \quad \alpha_i \in W_{h(i-1)+1}.
\]
For each $i = 1,2,\dotsc,g$ and $k \le n(i)$ we define the projections $p_i^k\colon G_i \to W_{h(i-1)+1}$ by
\[
    p_i^k( [\alpha_1, \alpha_2, \dotsc, \alpha_{n(i)}] ) = \alpha_k.
\]

\subsection{Complex Words}

Let $S_i \subseteq W_{h(i)}$. A \emph{word in $S_i$} is a product of the form
\[
    w = A_1^{\sigma(1)} \, A_2^{\sigma(2)} \, \dotsm A_\ell^{\sigma(\ell)}
\]
such that $A_j \in S_i$ and $\sigma(j) \in \{ -1, 1 \}$ for each $j$. We say that $\ell$ is the \emph{length} of the word. The word $w$ \emph{vanishes} if we have $w = I_{h(i)}$ in group $W_{h(i)}$.

The word is \emph{reduced} if $A_i = A_{i+1}$ implies $\sigma(i) = \sigma(i+1)$. That is, in a reduced word we do not have a subword of the form $B B^{-1}$ or $B^{-1} B$, $B \in S_i$. In the Cayley graph $\Cay(W_{h(i)}, S_i)$, a walk of length $\ell$ corresponds to a word of length $\ell$, a non-backtracking walk corresponds to a reduced word, and a cycle corresponds to a reduced walk that vanishes.

Let $X = \{ A_1, A_2, \dotsc, A_\ell \}$. We say that the \emph{complexity} of the word $w$ is $\ell - |X|$. For example, if $L_1$, $L_2$, and $L_3$ are distinct elements of $S_i$, then
\begin{enumerate}[noitemsep]
    \item the complexity of $L_1 L_2 L_3$ and $L_2^{-1} L_3$ is $0$,
    \item the complexity of $L_1 L_2 L_1 L_3$, $L_3 L_3$, and $L_1 L_3 L_1^{-1}$ is $1$,
    \item the complexity of $L_1 L_1 L_2 L_2$ and $L_1 L_3 L_1^{-1} L_1^{-1}$ is $2$.
\end{enumerate}
Obviously, the complexity of any word of length $g$ is at most $g-1$.

We say that a word $w$ is \emph{$k$-complex} if it is reduced and its complexity is at most $k$. For example, a $2$-complex word is always a $3$-complex word as well.

\subsection{Overview}

We will construct sets $S_1, S_2, \dotsc, S_g$ such that
\begin{enumerate}[noitemsep]
    \item $S_i \subseteq G_i$,
    \item $|S_i| = n(i)$,
    \item $(i-1)$-complex words in $S_i$ do not vanish.
\end{enumerate}
It follows that
\begin{enumerate}[noitemsep]
    \item $S_g \subseteq G_g$,
    \item $|S_g| = 2^m$,
    \item reduced words of length $1 \le \ell \le g$ in $S_g$ do not vanish.
\end{enumerate}
It follows that the girth of $\Cay(W_{h(g)}, S_g)$ is larger than $g$. In particular, Theorem~\ref{thm:high-girth} follows by choosing $g = 2r+1$ and $m > \log_2 k$. Our construction satisfies
\[
    g + m > \sqrt{h(g)} > \sqrt{\log_2 \log_2 |W_{h(g)}|}.
\]

\subsection{Base Case}

\paragraph{Construction.}

The base case of our construction, $S_1 \subseteq G_1$ is straightforward. Note that we have $h(0) + 1 = 1$; hence $G_1 = W_1^{n(1)}$. Recall that $W_1 = \{ I, X(1) \}$. We define
\begin{align*}
    L_j &= X_{h(1)}(j; n(1)), \\
    S_1 &= \{ L_1, L_2, \dotsc, L_{n(1)} \}.
\end{align*}
Observe that we have the projections
\begin{alignat*}{2}
    p_i^k(L_j) = p_i^k(L_j^{-1}) &= I, & k &\ne j, \\
    p_i^k(L_j) = p_i^k(L_j^{-1}) &= X(1), & \qquad k &= j.
\end{alignat*}

\paragraph{Correctness.}

Now assume that $w$ is a $0$-complex word in $S_1$. Then there is at least one $L_k \in S_1$ that occurs in $w$; moreover, $L_k$ occurs only once. It follows that $p_1^k(w) = X(1) \ne I$, and hence $w$ does not vanish in $G_1$.

\subsection{Recursive Step}\label{app:rec-step}

\paragraph{Construction.}

Let $i \ge 2$. Assume that we have already constructed a set $S_{i-1}$ such that $(i-2)$-complex words in $S_{i-1}$ do not vanish in $W_{h(i-1)}$. We proceed to construct $S_i \subseteq W_{h(i)}$.

We have $|S_{i-1}| = 2n(i)$; hence we can label the elements of $S_{i-1}$ by $\gamma_j$ and $\delta_j$ where $j = 1,2,\dotsc,n(i)$. We define
\begin{align*}
    K_j &= [ \gamma_j,\, \delta_j,\, \gamma_j,\, \delta_j,\, \dotsc,\, \gamma_j,\, \delta_j ], \\
    L_j &= K_j \, X_{h(i)}(j; n(i)), \\
    S_i &= \{ L_1, L_2, \dotsc, L_{n(i)} \}.
\end{align*}
Observe that we have the projections
\begin{alignat*}{2}
    p_i^k(K_j) &= [ \gamma_j, \delta_j ], \\
    p_i^k(L_j) &= [ \gamma_j, \delta_j ], & k &\ne j, \\
    p_i^k(L_j) &= [ \gamma_j, \delta_j ] \, X(1), & \qquad k &= j, \\
    p_i^k(L_j^{-1}) &= [ \gamma_j^{-1}, \delta_j^{-1} ], & k &\ne j, \\
    p_i^k(L_j^{-1}) &= X(1) \, [ \gamma_j^{-1}, \delta_j^{-1} ], & \qquad k &= j.
\end{alignat*}

\paragraph{Intuition.}

We will prove that a $(i-1)$-complex word $w$ in $S_i$ does not vanish in $W_{h(i)}$. We will argue that there is at least one projection $p_i^k(w) \ne I$. To gain some intuition, let us begin with some examples.

First, let $w = L_1 L_2 L_3 L_1$. Now
\begin{align*}
    p_i^1(w)
    &= [ \gamma_1, \delta_1 ] \, X(1) \, [ \gamma_2, \delta_2 ] \, [ \gamma_3, \delta_3 ] \, [ \gamma_1, \delta_1 ] \, X(1) \\
    &= [ \gamma_1 \delta_2 \delta_3 \delta_1,\, \delta_1 \gamma_2 \gamma_3 \gamma_1 ], \\
    p_i^2(w) 
    &= [ \gamma_1, \delta_1 ] \, [ \gamma_2, \delta_2 ] \, X(1) \, [ \gamma_3, \delta_3 ] \, [ \gamma_1, \delta_1 ] \\
    &= [ \gamma_1 \gamma_2 \delta_3 \delta_1,\, \delta_1 \delta_2 \gamma_3 \gamma_1 ] \, X(1), \\
    p_i^3(w)
    &= [ \gamma_1, \delta_1 ] \, [ \gamma_2, \delta_2 ] \, [ \gamma_3, \delta_3 ] \, X(1) \, [ \gamma_1, \delta_1 ] \\
    &= [ \gamma_1 \gamma_2 \gamma_3 \delta_1,\, \delta_1 \delta_2 \delta_3 \gamma_1 ] \, X(1), \\
    p_i^4(w)
    &= [ \gamma_1, \delta_1 ] \, [ \gamma_2, \delta_2 ] \, [ \gamma_3, \delta_3 ] \, [ \gamma_1, \delta_1 ] \\
    &= [ \gamma_1 \gamma_2 \gamma_3 \gamma_1,\, \delta_1 \delta_2 \delta_3 \delta_1 ].
\end{align*}
In particular, $w$ is a word of length $4$ in $S_i$, and it follows that all projections are of form $[ w_1, w_2 ] \, X(x)$, where $w_1$ and $w_2$ are words of length $4$ in $S_{i-1}$. Moreover, as the original word $w$ was $1$-complex, $w_1$ and $w_2$ are also $1$-complex: each $L_j$ contributes precisely one $\delta_j$ or $\gamma_j$ in $w_1$ and $w_2$. Projection $p_i^4(w)$ is not useful for our purposes; nothing interesting happens there, as $w$ did not contain any copies of $L_4$. However, the case of $p_i^1(w)$ is more interesting: we have $p_i^1(w) = [ w_1, w_2 ]$, and this time $w_1$ and $w_2$ are not only $1$-complex but also $0$-complex.

Second, let $w = L_2 L_1 L_2^{-1}$. Now
\begin{align*}
    p_i^1(w)
    &= [ \gamma_2, \delta_2 ] \, [ \gamma_1, \delta_1 ] \, X(1) \, [ \gamma_2^{-1}, \delta_2^{-1} ] \\
    &= [ \gamma_2 \gamma_1 \delta_2^{-1}, \, \delta_2 \delta_1 \gamma_2^{-1} ].
\end{align*}
Again we were able to identify a projection $p_i^1(w) = [ w_1, w_2 ]$ such that $w_1$ and $w_2$ are $0$-complex even though $w$ is $1$-complex.

\paragraph{Correctness.}

To generalise the above observations, we study the following cases that cover all possible $(i-1)$-complex words in $S_i$.
\begin{enumerate}
    \item For some $k$, word $w$ is of the form
    \[
        w = s L_k t L_k u,
    \]
    where $t$ does not contain any copies of $L_k$ or $L_k^{-1}$. It follows that we have
    \begin{align*}
        p_i^k(s) &= X(s') \, [ s_1, s_2 ], &
        p_i^k(t) &= [ t_1, t_2 ], &
        p_i^k(u) &= [ u_1, u_2 ] \, X(u')
    \end{align*}
    for some $s_1, s_2, t_1, t_2, u_1, u_2 \in W_{h(i-1)}$ and $s', u' \in \{0,1\}$. Hence
    \begin{align*}
        p_i^k(w) &= X(s') \, [ s_1, s_2 ] \, [ \gamma_k, \delta_k ] \, X(1) \,
                             [ t_1, t_2 ] \, [ \gamma_k, \delta_k ] \, X(1) \,
                             [ u_1, u_2 ] \, X(u') \\
        &= X(s') \, [ s_1 \gamma_k t_2 \delta_k u_1,\,
                      s_2 \delta_k t_1 \gamma_k u_2 ] \, X(u').
    \end{align*}
    Now if $w$ is an $(i-1)$-complex word in $S_i$, then $s_1 \gamma_k t_2 \delta_k u_1$ is $(i-2)$-complex in $S_{i-1}$: we have removed at least one duplicate element, as one slot of $L_k$ is replaced with $\gamma_k$ and while the other slot is replaced with $\delta_k$. Hence $p_i^k(w) \ne I$ and $w$ does not vanish.
    \item For some $k$, word $w$ is of the form $w = s L_k^{-1} t L_k^{-1} u$, where $t$ does not contain any copies of $L_k$ or $L_k^{-1}$. It follows that $w^{-1} = u^{-1} L_k t^{-1} L_k u^{-1}$, and the above argument shows that $w^{-1} \ne I$. Hence $w$ does not vanish.
    \item For some $j \ne k$ and some $a,b,c \in \{-1,1\}$, word $w$ is of the form
    \[
        w = s \, L_k^{a} \, t \, L_j^{b} \, u \, L_k^{c} \, v,
    \]
    where $t$ and $u$ do not contain any copies of $L_j$, and $L_j^{-1}$. It follows that
    \begin{align*}
        p_i^j(s) &= X(s') \, [ s_1, s_2 ], &
        p_i^j(t) &= [ t_1, t_2 ], \\
        p_i^j(u) &= [ u_1, u_2 ], &
        p_i^j(v) &= [ v_1, v_2 ] \, X(v')
    \end{align*}
    for some $s_1, s_2, t_1, t_2, u_1, u_2, v_1, v_2 \in W_{h(i-1)}$ and $s', v' \in \{0,1\}$. Hence
    \begin{align*}
        p_i^j(w) &= X(s') \, [ s_1, s_2 ] \, [ \gamma_k^{a}, \delta_k^{a} ] \,
                             [ t_1, t_2 ] \, [ \gamma_j^{b}, \delta_j^{b} ] \, X(1) \,
                             [ u_1, u_2 ] \, [ \gamma_k^{c}, \delta_k^{c} ] \,
                             [ v_1, v_2 ] \, X(v') \\
        &= X(s') [ s_1 \, \gamma_k^{a} \, t_1 \, \gamma_j^{b} \, u_2 \, \delta_k^{c} \, v_2,\,
                   s_2 \, \delta_k^{a} \, t_2 \, \delta_j^{b} \, u_1 \, \gamma_k^{c} \, v_1 ] \, X(v' + 1).
    \end{align*}
    Again, if $w$ is an $(i-1)$-complex word in $S_i$, then $s_1 \, \gamma_k^{a} \, t_1 \, \gamma_j^{b} \, u_2 \, \delta_k^{c} \, v_2$ is $(i-2)$-complex in $S_{i-1}$: on slot of $L_k$ contributes a $\gamma_k$ while the other slot contributes a $\delta_k$. Hence $w$ does not vanish.
    \item The above cases cover all reduced words $w$ that contain at least two occurrences of any element of $S_i$; recall that in a reduced word, we cannot have $L_k$ and $L_k^{-1}$ next to each other. The only remaining case is that all elements of $w$ are distinct. Then $w$ is $0$-complex, and $p_i^j(w) = [ w_1, w_2 ] \, X(x)$, where $w_1$ and $w_2$ are $0$-complex. In particular, they are $(i-2)$-complex, they do not vanish, and hence $w$ does not vanish, either.
\end{enumerate}
This concludes the proof of Theorem~\ref{thm:high-girth}. \qed

\begin{figure}
    \centering
    \includegraphics[page=\PHighGirth]{figs.pdf}
    \caption{Construction of Appendix~\ref{app:high-girth}, in the case $g = 3$ and $m = 1$. We have $h(0) = 0$, $h(1) = 4$, $h(2) = 7$, and $h(3) = 9$. There are $n(1) = 8$ elements in $S_1$; these can be interpreted as complete binary trees of height $4$, as they are elements of group $W_4$. In the illustration, a black internal node indicates the value $1$ and a white internal node indicates the value $0$. For the purposes of constructing $S_2 = \{ L_1,L_2,L_3,L_4 \}$ in Appendix~\ref{app:rec-step}, we label the elements of $S_1$ by $\gamma_1, \delta_1, \gamma_2, \delta_2, \dotsc, \gamma_4, \delta_4$. Now $L_j$ is constructed by gluing together copies of $\gamma_j$ and $\delta_j$.}\label{fig:high-girth}
\end{figure}

\end{document}